\documentclass[a4paper,UKenglish,final]{lipics}

\usepackage{color} 
\usepackage{amsmath,amssymb} 
\usepackage{enumerate}
\usepackage[basic]{complexity}
\usepackage[T1]{fontenc}
\usepackage{tikz}
\usepackage[obeyFinal]{todonotes}
\usepackage{xspace}
\usepackage{paralist}
\usepackage{graphicx}
\usepackage{stmaryrd}
\usepackage{wrapfig}
\usepackage[strict]{changepage}
\usepackage{caption}

\title{Quantitative Games under Failures}

\author[1]{Thomas Brihaye}
\author[2]{Gilles Geeraerts}
\author[1]{Axel Haddad}
\author[2]{Benjamin Monmege}
\author[2]{Guillermo A. P\'erez\thanks{Author supported by F.R.S.-FNRS
	fellowship.}}
\author[1]{Gabriel Renault}

\affil[1]{Universit\'e de Mons, Belgium,
\texttt{\{thomas.brihaye,axel.haddad,gabriel.renault\}@umons.ac.be}}
\affil[2]{Universit\'e libre de Bruxelles, Belgium,
\texttt{\{gigeerae,benjamin.monmege,gperezme\}@ulb.ac.be}}

\authorrunning{T. Brihaye, G. Geeraerts, A. Haddad, B.
Monmege, G. A. P\'erez, G. Renault}

\Copyright{Thomas Brihaye, Gilles Geeraerts, Axel Haddad, Benjamin
Monmege, Guillermo A. P\'erez, \\ \phantom{a}~~~~~~~~~~~~~~~~~Gabriel Renault}

\serieslogo{}
\volumeinfo
  {Billy Editor and Bill Editors}
  {2}
  {Conference title on which this volume is based on}
  {1}
  {1}
  {1}
\EventShortName{}
\DOI{10.4230/LIPIcs.xxx.yyy.p}

\begin{document}


\usetikzlibrary{positioning,fit,arrows,automata,calc,shapes} 
\tikzset{-latex,>=stealth',shorten >=1pt,shorten <=1pt,auto,node distance=1cm,
every loop/.style={looseness=6},
initial text={},
every fit/.style={draw,densely dotted,rectangle},
el/.style={font=\scriptsize},
squircle/.style={draw,rounded corners=2pt},
inner sep=1mm,
loopright/.style={loop,looseness=6,out=-45, in=45},
loopleft/.style={loop,looseness=6,out=135, in=225},
loopabove/.style={loop,looseness=6,out=45, in=135},
loopbelow/.style={loop,looseness=6,out=-135, in=-45},
}

\newtheorem{proposition}[theorem]{Proposition}
\newtheorem{claim}[theorem]{Claim}

\providecommand\scalemath{}
\renewcommand\scalemath[2]{\scalebox{#1}{\mbox{\ensuremath{\displaystyle #2}}}}

\renewclass{\AP}{APTIME}
\renewclass{\P}{PTIME}
\renewclass{\EXP}{EXPTIME}

\newcommand{\adam}{\textnormal{Maximiser}\xspace}
\newcommand{\eve}{\textnormal{Minimiser}\xspace}
\newcommand{\adamshort}{\textit{Max}\xspace}
\newcommand{\eveshort}{\textit{Min}\xspace}

\newcommand{\Min}{\ensuremath{\mathrm{Min}}\xspace}
\newcommand{\Max}{\ensuremath{\mathrm{Max}}\xspace}
\newcommand{\outcomes}{\textit{Play}\xspace}
\def\saboteur{\textnormal{Saboteur}\xspace}
\def\runner{\textnormal{Runner}\xspace}
\def\QSG{{\sf QSG}\xspace}
\def\QSGs{{\sf QSG}s\xspace}

\newcommand{\ThPr}{\textit{ThPr}\xspace}
\newcommand{\SPr}{\textit{SPr}\xspace}
\newcommand{\ESPr}{\textit{ESPr}\xspace}

\newcommand\uppervalue{\overline{\Val}}
\newcommand\lowervalue{\underline{\Val}}
\newcommand\Parity{\textit{Parity}}
\newcommand\Col{Col}

\newcommand{\Veve}{V_{\eveshort}}
\newcommand{\Vadam}{V_{\adamshort}}
\newcommand{\VMin}{V_{\Min}}
\newcommand{\VMax}{V_{\Max}}
\newcommand{\VSem}{V^{c}}
\newcommand{\wSem}{w}
\newcommand{\VMinSem}{V_{\Min}^{c}}
\newcommand{\VMaxSem}{V_{\Max}^{c}}
\newcommand{\edgesSem}{E^{c}}
\newcommand{\InitSem}{v_I^{c}}
\newcommand{\Val}{\ensuremath{\mathbf{Val}}}
\newcommand{\Valstatic}{\ensuremath{\mathbf{Val}_{\mathrm{stat}}}\xspace}
\newcommand{\supfun}{\ensuremath{\mathsf{Sup}}\xspace}
\newcommand{\inffun}{\ensuremath{\mathsf{Inf}}\xspace}
\newcommand{\lsupfun}{\ensuremath{\mathsf{LimSup}}\xspace}
\newcommand{\linffun}{\ensuremath{\mathsf{LimInf}}\xspace}
\newcommand{\mpfun}{\ensuremath{\mathsf{Avg}}\xspace}
\newcommand{\discfun}[1]{\ensuremath{\mathsf{DS}_{#1}}\xspace}
\newcommand{\dsfun}{\ensuremath{\mathsf{DS}}\xspace}
\newcommand{\StratAdam}{\Sigma_{\adamshort}}
\newcommand{\StratEve}{\Sigma_{\eveshort}}
\newcommand{\StratMax}{\Sigma_{\Max}}
\newcommand{\StratMin}{\Sigma_{\Min}}
\newcommand{\stratMax}{\sigma}
\newcommand{\stratMin}{\rho}
\newcommand{\StratSaboteur}{\Sigma_{\mathrm{Sab}}}
\newcommand{\StratRunner}{\Sigma_{\mathrm{Run}}}
\newcommand{\Plays}{\ensuremath{\mathrm{Plays}}\xspace}
\newcommand{\Prefs}{\ensuremath{\mathrm{Prefs}}\xspace}
\newcommand{\PrefsSaboteur}{\ensuremath{\mathrm{Prefs}_{\mathrm{Sab}}}\xspace}
\newcommand{\PrefsRunner}{\ensuremath{\mathrm{Prefs}_{\mathrm{Run}}}\xspace}
\newcommand{\distr}{\delta}

\newcommand{\valuation}{\mathrm{val}}
\newcommand{\nplayer}{\mathrm{Player}}

\newcommand\rouge[1]{{\color{red} #1}}
\newcommand\defin[1]{\textit{\textbf{#1}}}

\newcommand\N{\mathbb{N}}
\newcommand\Z{\mathbb{Z}}
\renewcommand\R{\mathbb{R}}

\newcommand\graph{G}
\newcommand\edges{E}
\newcommand\vertices{V}
\newcommand\weights{w}
\newcommand\game{\ensuremath{\mathcal{G}}\xspace}
\newcommand\egame{\ensuremath{\mathcal{E}}\xspace}
\def\sem#1{\ensuremath{\llbracket#1\rrbracket}}
\newcommand\gameSem{\sem{\game}\xspace}
\newcommand\gameStat[1]{\mathcal{G}_{#1}}
\newcommand\budget{B}
\newcommand\buDistA{BD_1}
\newcommand\buDistB{BD_2}
\newcommand\buRel{\triangleright}
\newcommand\odeg{\ensuremath{ \textsf{deg}^+ }}
\renewcommand\st{\:\mid\:}

\newcommand\clause{\textit{Cl}}
\newcommand\Ver{\textit{Ver}}

\renewcommand\ge{\geqslant}
\renewcommand\geq{\geqslant}
\renewcommand\le{\leqslant}
\renewcommand\leq{\leqslant}

\maketitle

\begin{abstract}
  We study a generalisation of sabotage games, a model of dynamic
  network games introduced by van Benthem~\cite{Van02}. The original
  definition of the game is inherently finite and therefore does not
  allow one to model infinite processes. We propose an extension of
  the sabotage games in which the first player (\runner) traverses an
  arena with dynamic weights determined by the second player
  (\saboteur). In our model of \emph{quantitative sabotage games},
  \saboteur is now given a budget that he can distribute amongst the
  edges of the graph, whilst \runner attempts to minimise the quantity
  of budget witnessed while completing his task. We show that, on the
  one hand, for most of the classical cost functions considered in the
  literature, the problem of determining if \runner has a strategy to
  ensure a cost below some threshold is \EXP-complete. On the other
  hand, if the budget of \saboteur is fixed a priori, then the problem
  is in $\P$ for most cost functions. Finally, we show that
  restricting the dynamics of the game also leads to better
  complexity.
\end{abstract}

\clearpage


\section{Introduction}

Two-player games played on graphs are nowadays a well-established
model for systems where two antagonistic agents interact. In
particular, they allow one to perform controller
synthesis~\cite{ag11}, when one of the players models the controller,
and the second plays the role of an evil environment. Quantitative
generalisations (played on weighted graphs) of these models have attracted
much attention in the last decades~\cite{em79,jurdzinski98,cdh10} as
they allow for a finer analysis of those systems.

In this setting, most results assume that the arena (i.e., the graph)
on which the game is played does not change during the game. There are
however many situations where this restriction is not natural, at
least from a modelling point of view. For instance, Gr{\"{u}}ner
\textit{et al.} \cite{grt13} model connectivity problems in
\emph{dynamic} networks (i.e., subject to failure and restoration)
using a variant of \emph{sabotage games} -- a model originally proposed
by van Benthem \cite{Van02} -- to model \emph{reachability problems} in a
network prone to errors. A sabotage game is played on a directed
graph, and starts with a token in an initial vertex. Then, \runner and
\saboteur (the two players of the game) play in alternation: \runner
moves the token along one edge and \saboteur is allowed to remove one
edge. \runner wins the game if he reaches a target set of
vertices. In~\cite{LodRoh03}, it is shown that deciding the existence
of a winning strategy for \runner is \PSPACE-complete.

In those sabotage games, errors are regarded as unrecoverable
failures. In practice, this hypothesis might be too strong. Instead,
one might want to model the fact that certain uncontrollable events
incur additional costs (modelling delays, resource usage\ldots), and
look for strategies that allow one to fulfil the game objective
\emph{at a minimal cost}, whatever the occurrence of uncontrollable
events. For instance, if the graph models a railway network, the
failure of a track will eventually be fixed, and, in the meantime,
trains might be slowed down on the faulty portion or diverted,
creating delays in the journeys. It is thus natural to consider
\emph{quantitative} extensions of sabotage games, where \saboteur
controls the price of the actions in the game. This is the aim of the
present paper.

More precisely, we extend sabotage games in two directions. First, we
consider games played on \emph{weighted} graphs. \saboteur is allotted
an integral budget $B$ that he can distribute (dividing it into
integral parts) on the edges of the graph, thereby setting their
weights. At each turn, \saboteur can change this distribution by
moving $k$ units of budget from an edge to another edge (for
simplicity, we restrict ourselves to $k=1$ but our results hold for
any $k$). Second, we relax the inherent finiteness of sabotage games
(all edges will eventually be deleted), and consider infinite horizon
games (i.e., plays are now infinite).  In this setting, the goal of
\runner is to minimise the cost defined by the sequence of weights of
edges visited, with respect to some fixed cost function ($\inffun$,
$\supfun$, $\linffun$, $\lsupfun$, average or discounted-sum), while
\saboteur attempts to maximise the same cost. We call these games
\emph{quantitative sabotage games} (\QSG, for short).

Let us briefly sketch one potential application of our model, showing
that they are useful to perform synthesis in a dynamic environment. Our
application is borrowed from Suzuki and Yamashita~\cite{ys10} who have
considered the problem of \emph{motion planning} of multiple mobile
robots that interact in a finite space. In essence, each robot
executes a ``Look-Compute-Move'' cycle and should realise some
specification (that we could specify using LTL, for instance). For
simplicity, assume that at every observation (Look) phase, at most one
other robot has moved.  Clearly every motion phase (Move) will require
different amounts of time and energy depending on the location of the
other robots. We can model the interaction of each individual robot
against all others using a \QSG where \runner is one robot, \saboteur
is the coalition of all other robots, and the budget is equal to the
number of robots minus $1$. This model allows one to answer meaningful
questions such as `\emph{what is, in the worst case, the average delay
  the robot incurs because of the dynamics of the system?}', or 
`\emph{what is the average amount of additional energy
required because of the movements of the other robots?}' using
appropriate cost functions.

As a second motivational example, let us recall the motivation of the original
Sabotage Game: consider a situation in which you need to find your way between
two cities within a railway network where a malevolent demon starts cancelling
connections? This is called the \emph{real Travelling Salesman Problem} by
Benthem~\cite{Van02}. However, in real life, railway companies have contracts
with infrastructure companies which ensure that failures in the railway network
are repaired withing a given amount of time (e.g. a service-level agreement). In
this case, it is better to consider delays instead of absolute failures in the
network. Further, salesmen do not usually have one single trip in their whole
carriers. For modelling purposes, one can in fact assume they never stop
travelling. In this setting, \QSGs can be used to answer the question:
`\emph{what is, in the worst case, the average delay time incurred by the
salesman\emph}'? Our model can be used to treat the same questions for other
networks and not just railway networks.

\subparagraph{Related Works \& Contributions.}  Variations of the
original sabotage games have been considered by students of van
Benthem. In~\cite{kurzen11}, the authors have considered changing the
\emph{reachability objective} of \runner to a \emph{safety objective},
and proved it is \PSPACE-complete as well. They also consider a
co-operative variation of the game which, not surprisingly, leads to a
lower complexity: \NL-complete.  In~\cite{sevenster06}, an asymmetric
imperfect information version of the game is studied---albeit, under
the guise of the well-known parlor game \emph{Scotland Yard}---and
shown to be \PSPACE-complete.  We remark that although the latter
version of sabotage games already includes some sort of dynamicity in
the form of the \emph{Scotland Yard} team moving their pawns on the
board, both of these studies still focus on inherently finite versions
of the game.

We establish that \QSGs are \EXP-complete in general. Our approach is to prove
the result for a very weak problem on \QSGs, called the \emph{safety problem},
that asks whether \runner can avoid \emph{ad vitam \ae{}ternam} edges with
non-zero budget on it. We remark that although the safety problem is related to
cops and robbers games~\cite{ag11,gr95}, we were not able to find \EXP-hard
variants that reduce easily into our formalism.\footnote{We compare to related
	works on cops and robbers games in Appendix~\ref{CopsAndRobbers}.}
The general problem being \EXP-complete, we consider the case
where the budget is fixed instead of left as an input of the
problem (see Corollary~\ref{cor:fixed}). We also consider restricting
the behaviour of \saboteur and define a variation of our \QSGs in
which \saboteur is only allowed to choose an initial distribution of
weights but has to commit to it once he has fixed it. We call this the
\emph{static} version of the game. For both restrictions, we show that
tractable algorithms exist for some of the cost functions we
consider.
A summary of the complexity results we establish in this work is shown
in Table~\ref{tbl:summary}. In Section~\ref{sec:conclusions}, we
comment on several implications of the complexity bounds proved in
this work.

\begin{table}[tbp]
  \caption{Complexity results for quantitative sabotage games}
  \label{tab:results}
  \centering
  \begin{tabular}{|l|c|c|c|}
    \hline
    & \QSG & static \QSG & fixed budget \QSG\\\hline\hline
    \inffun,~\linffun & $ \in \EXP$ & $\in\P$ & $\in\P$ \\\hline
    \supfun,~\lsupfun,~\mpfun & \EXP-c & \coNP-c & $\in\P$  \\\hline
    \dsfun & \EXP-c & \coNP-c & $\in \NP\cap\coNP$ \\\hline
  \end{tabular}
  \label{tbl:summary}
\end{table}


\section{Quantitative sabotage games}

Let us now formally define quantitative sabotage games (\QSG). We start with
the definition of the cost functions we will consider, then give the
syntax and semantics of \QSG.

\subparagraph{Cost functions.} A \emph{cost function}
$f\colon \mathbb{Q}^\omega \to \mathbb{R}$ associates a real number to
a sequence of rationals $u=(u_i)_{i\ge 0}\in\mathbb{Q}^\omega$.
The six classical cost functions that we
consider are
\begin{itemize}
	\item $\inffun(u) = \inf\{ u_i \st i \ge 0\}$;
	\item $\supfun(u) = \sup\{ u_i \st i \ge 0\}$;
	\item $\linffun(u) = \liminf_{n \rightarrow \infty}\{ u_i \st i \ge n\}$;
	\item $\lsupfun(u) = \limsup_{n \rightarrow \infty}\{ u_i \st i \ge n\}$;
	\item $\mpfun(u) = \liminf_{n \to \infty} \frac{1}{n} \sum_{i = 0}^n u_i$,
		which stands for the average cost (also called \emph{mean-payoff}
		in the literature); and
	\item $\discfun{\lambda}(u) = \sum_{i=0}^\infty \lambda^i \cdot u_i$,
		(with $0<\lambda<1$), stands for discounted-sum.
\end{itemize}
In the following, we let $\dsfun=\{ \discfun{\lambda} \mid 0<\lambda<1\}$.

\subparagraph{Syntax.} As sketched in the introduction,
\emph{quantitative sabotage games} are played by \runner and \saboteur
on a directed weighted graph, called the \emph{arena}.  A~play
alternates between \runner moving the token along the edges and
\saboteur modifying the weights.  We consider that \saboteur has a
fixed integer budget~$B$ that he can distribute on edges, thereby
setting their weights (which must be integer values).  Formally, for a
finite set $E$
and a budget $B\in \mathbb N$, 
$\Delta(E,B)$ denotes the set of all \emph{distributions} of budget
$B$ on $E$, where a distribution is a function
$\distr\colon E \to \{0,1,\ldots,B\}$ such that
$\sum_{e \in E} \distr(e) \le B$ (the last constraint is an inequality
since the whole budget need not be distributed on $E$). Then, a
\emph{quantitative sabotage game} is a tuple
$\game = (V, E, B, v_I, \distr_I, f)$, where $(V,E)$ is a directed
graph, $B \in \mathbb{N}$ is the budget of the game, $v_I \in V$ is
the initial vertex, $\distr_I\in\Delta(E,B)$ is the initial
distribution of the budget, and $f$ is a cost function.  We assume,
without loss of generality, that there are no deadlocks in $(V,E)$,
i.e., for all $v\in V$, there is $v'\in V$ such that $(v,v')\in E$.
In the following, we may alternatively write $\Delta(\game)$ for
$\Delta(E,B)$ when $\game$ is a \QSG with set of edges $E$ and budget
$B$.

\subparagraph{Semantics.} To define the semantics of a \QSG $\game$, we
first formalise the possible redistributions of the budget by
\saboteur. We choose to restrict them, reflecting some physical
constraints: \saboteur can move at most one unit of weight in-between
two edges.  For $\distr,\distr' \in \Delta(\game)$, we say that
$\distr'$ is a \emph{valid redistribution} from $\distr$, noted
$\distr \buRel \distr'$, if and only if there are $e_1,e_2\in E$ such
that $\distr'(e_1) \in \{\distr(e_1), \distr(e_1) - 1\}$,
$\distr'(e_2) \in \{\distr(e_2), \distr(e_2) + 1\}$, and for all other
edges $e\not\in \{e_1,e_2\}$, $\distr'(e)=\distr(e)$.  Then, a
\emph{play} in a \QSG $\game = (V, E, B, v_I, \distr_I, f)$ is an
infinite sequence $\pi=v_0 \distr_0 v_1 \distr_1 \cdots$ alternating
vertices $v_i\in V$ and budget distributions
$\distr_i\in \Delta(\game)$ such that
\begin{inparaenum}[$(i)$]
\item $v_0 = v_I$;
\item  $\distr_0 = \distr_I$; and
\item for all $i\ge 0$: $(v_i,v_{i+1}) \in E$, and
  $\distr_i \buRel \distr_{i+1}$.
\end{inparaenum}
Let 
$\Prefs_\Delta(\game)$ denote the set of prefixes of plays ending in a
budget distribution, and $\Prefs_V(\game)$ the set of prefixes of
length at least 2 ending in a vertex. We abuse notations and lift cost
functions $f$ to plays letting
$f(v_0 \distr_0 v_1
\distr_1\cdots)=f(\distr_0(v_0,v_1)\distr_1(v_1,v_2)\cdots)$.
A \emph{strategy} of \runner is a mapping
$\rho\colon \Prefs_\Delta(\game)\to V$ such that
$(v_n,\rho(\pi))\in E$ for all
$\pi=v_0 \distr_0 \cdots v_{n} \distr_n\in \Prefs_\Delta(\game)$.  A
strategy of \saboteur is a mapping
$\sigma\colon\Prefs_V(\game)\to \Delta(\game)$ such that
$\distr_{n-1}\buRel\sigma(\pi)$ for all
$\pi= v_0 \distr_0 \cdots v_{n-1}\distr_{n-1}v_n\in\Prefs_V(\game)$.
We denote by $\StratRunner(\game)$ (respectively,
$\StratSaboteur(\game)$) the set of all strategies of \runner
(respectively, \saboteur).  A pair of strategies $(\rho,\sigma)$ of
\runner and \saboteur defines a unique play
$\pi_{\rho,\sigma} = v_0 \distr_0 v_1 \distr_1 \cdots$ such that for
all $i \ge 0$:
\begin{inparaenum}[$(i)$]
\item $v_{i+1}=\rho(v_0 \distr_0\cdots v_i \distr_i)$; and
\item $\distr_{i+1}=\sigma(v_0 \distr_0\cdots v_i \distr_iv_{i+1})$.
\end{inparaenum}

\subparagraph{Values and determinacy.} We are interested in computing the
best value that each player can guarantee no matter how the other
player plays. To reflect this, we define two values of a \QSG $\game$:
the superior value (modelling the best value for \runner)as
\(\overline\Val(\game) := 
\sup_{\stratMax \in \StratSaboteur(\game)}\inf_{\stratMin \in \StratRunner(\game)}
f(\pi_{\stratMin,\stratMax})\),
and the inferior value (modelling the best value for Sabo-teur)
as \(\underline\Val(\game) := \inf_{\stratMin \in \StratRunner(\game)}
\sup_{\stratMax \in \StratSaboteur(\game)}
f(\pi_{\stratMin,\stratMax})\).
It is folklore to prove that
$\underline\Val(\game)\leq \overline\Val(\game)$. Indeed, for the
previously mentioned cost functions, we can prove that \QSGs are
determined, i.e., that $\underline\Val(\game)=\overline\Val(\game)$
for all \QSGs $\game$. This can be formally proved by encoding a \QSG
$\game$ into a quantitative two-player game $\sem\game$ (whose
vertices contain both vertices of $\game$ and budget distributions),
and then using classical Martin's determinacy theorem~\cite{Mar75}, as
formally done in Appendix~\ref{sec:qgames-sem}.
$\underline\Val(\game)=\overline\Val(\game)$ is henceforth called the
\emph{value of $\game$}, and denoted by $\Val(\game)$.

\def\nd#1{\tikz[baseline=-.6ex]{\node[draw,rounded corners]{$#1$};}}
\begin{wrapfigure}[7]{r}{0.2\textwidth}
  \centering
  \begin{tikzpicture}
    \node[draw,rounded corners] (1) at (1,1) {$1$};
    \node[draw,rounded corners] (2) at (1,0) {$2$};
    \node[draw,rounded corners] (3) at (0,0.5) {$3$};
    
    \draw[-latex] (1) edge (2);
    \draw[-latex] (2) edge[bend right=40] (1);
    \draw[-latex] (2) edge (3);
    \draw[-latex] (3) edge (1);
  \end{tikzpicture}
  \caption{A \QSG}
  \label{fig:example}
\end{wrapfigure} 
\subparagraph{Example.} Consider the simple \QSG $\game$ in
\figurename~\ref{fig:example}, where the budget of \saboteur is $B=4$,
and the cost function is $\mpfun$. We claim that whatever the initial
configuration, $\Val(\game)=2$. Indeed, consider the strategy of
\saboteur that consists in eventually putting all the budget on the
edge $(\nd{1},\nd{2})$ (i.e., letting $\distr(\nd{1},\nd{2})=4$ and
$\distr(e)=0$ for all other edges $e$), and then playing as follows:
whenever \runner reaches $\nd{2}$, move one unit of budget from
$(\nd{1},\nd{2})$ to $(\nd{2},\nd{3})$; if \runner moves to $\nd{3}$,
move the unit of budget from $(\nd{2},\nd{3})$ to $(\nd{3},\nd{1})$;
and when \runner moves back to $\nd{1}$, move all the budget back on
$(\nd{1},\nd{2})$, by consuming one unit either from $(\nd{2},\nd{3})$
or from $(\nd{3},\nd{1})$.  Let us call this strategy
$\overline{\sigma}$. Since we consider the average cost, only the
long-term behaviour of \runner is relevant to compute the cost of a
play. So, as soon as \saboteur has managed to reach a distribution
$\distr$ such that $\distr(\nd{1},\nd{2})=4$, the only choices for
\runner each time he visits \nd{1} are either to visit the
\nd{1}--\nd{2}--\nd{3}--\nd{1} cycle, or the \nd{1}--\nd{2}--\nd{1}
cycle. In the former case, \runner traverses $3$ edges and pays
$4+1+1=6$, hence an average cost of $\frac{6}{3}=2$ for this cycle. In
the latter, he pays an average of $\frac{4+0}{2}=2$ for the
cycle. Hence, whatever the strategy $\rho$ of \runner, we have
$\mpfun(\pi_{\overline{\sigma},\rho})=2$, which proves that
$\underline\Val(\game)\geq 2$. One can check that the strategy
$\overline\rho$ of \runner consisting in always playing the
\nd{1}--\nd{2}--\nd{3}--\nd{1} cycle indeed guarantees cost $2$,
proving that $\overline\Val(\game)\leq 2$. This proves that the value
$\Val(\game)$ of the game is $2$.


\section{Solving quantitative sabotage games}\label{sec:solving}

Given a \QSG, our main objective is to determine whether \runner can
play in such a way that he will ensure a cost at most $T$, no matter
how \saboteur plays, and where $T$ is a given threshold. This amounts
to determining whether $\Val(\game) \leq T$.  Thus, for a cost function $f$,
the \textsc{Threshold problem with cost function $f$} consists in
determining whether $\Val(\game) \leq T$, given a \QSG $\game$ with
cost function $f$ and a non-negative threshold $T$.
When $f=\dsfun$, we assume that the discount factor $\lambda$ is part
of the input. If we want it to be a parameter of the problem (and not
a part of the input), we consider $f=\discfun{\lambda}$. Our main
contribution is to characterise the complexity of the threshold
problem for all the cost functions introduced before, as summarised in
the following theorem:
\begin{theorem}\label{thm:exptime}
  For cost functions $\supfun$, $\lsupfun$, $\mpfun$, $\dsfun$ and
  $\discfun{\lambda}$, the threshold problem over \QSG{s} is
  \EXP-complete; for $\inffun$ and $\linffun$, it is in \EXP.
\end{theorem}

For all cost functions, the $\EXP$ membership is established by using
the encoding (explained in Appendix~\ref{sec:qgames-sem}) of a \QSG
$\game$ into a classical quantitative two-player game $\sem\game$
which is played on a weighted graph, whose vertices are the
configurations of the sabotage game, i.e., a tuple containing the
current vertex, the last crossed edge and the current weight
distribution, and whose weights are in $\{0,\ldots,B\}$ (describing
how much runner pays by moving from one configuration to another).
Notice that $\Delta(\game)$ has size at most $(B+1)^{|E|}$, since
every distribution is a mapping of $E\to\{0,1,\ldots,B\}$.  Hence, we
see that the game $\sem\game$ has a number of vertices at most
exponential with respect to $|V|$, and polynomial with respect to $B$
(which, being given in binary, can be exponential in the size of the
input of the problem).  Using results from~\cite{ZwiPat96,cdh10,ag11},
we know that we can compute in pseudo-polynomial time the value of the
quantitative game $\sem\game$ for all the cost functions cited in the
theorem: here, pseudo-polynomial means polynomial with respect to the
number of vertices and edges of $\sem\game$ (which is exponential with
respect to $|V|$), and polynomial with respect to the greatest weight
in absolute value, here $B$ (which is also exponential with respect to
$|V|$).  Thus we obtain the exponential time upper bound announced in
the theorem.  Note that for $\discfun{\lambda}$, pseudo-polynomial
also means polynomial in the value of the denominator of
$\lambda$.\footnote{In case of discounted-sum, we design $\sem\game$
  with a discount factor $\sqrt\lambda$ (not necessarily rational),
  but we ensure that only one turn over two has a non-zero weight, so
  that we may indeed apply the reasoning of \cite{ZwiPat96} and their
  pseudo-polynomial algorithm.}

When the budget $B$ is fixed, i.e., when it is a parameter of the
problem and not one of the inputs, the explanation above can be
adapted to prove that the problem is solvable in polynomial time for
all but the $\discfun{\lambda}$ cost functions. 
Indeed, we can refine
our analysis of the size of $\Delta(\game)$. A budget distribution can
also be encoded as a mapping $\gamma\colon \{1,\ldots,B\} \to E$ where
we consider the budget as a set of indexed pebbles: such a mapping
represents the distribution $\distr$ defined by
$\distr(e)= |\gamma^{-1}(e)|$. This encoding shows that
$\Delta(\game)$ has size at most $|E|^B$, which is polynomial in
$|E|$. 
For the discounted sum, the role of $\lambda$ in the complexity
stays the same, causing an $\NP\cap\coNP$ and pseudo-polynomial
complexity: this blow-up disappears if $\lambda$ is a parameter of the
problem. In the overall, we obtain:

\begin{corollary}\label{cor:fixed}
  For cost functions $\inffun$, $\supfun$, $\linffun$, $\lsupfun$,
  $\mpfun$, $\discfun{\lambda}$, and for fixed budget $B$, the
  threshold problem for \QSGs is in $\P$; for $\dsfun$ (where
  $\lambda$ is an input), it is in $\NP\cap \coNP$ and can be solved
  in pseudo-polynomial time.
\end{corollary}

The rest of this section is devoted to the proof of $\EXP$-hardness in
Theorem~\ref{thm:exptime} for cost functions $\supfun$, $\lsupfun$,
$\mpfun$ and $\discfun{\lambda}$ (this implies $\EXP$-hardness for
$\dsfun$ too). Our gold-standard problem for \EXP-hardness is the
\emph{alternating Boolean formula} (ABF) problem, introduced by
Stockmeyer and Chandra in~\cite{SC79}. Our proof consists of a
sequence of reductions from this problem, as depicted in
\figurename~\ref{fig:reductions-dir}. First, we show a reduction to
the threshold problem for $\supfun$ cost function when the threshold
is $0$ and \textbf{the initial distribution is empty} (i.e., no budget
on any edge), on \QSGs extended with \emph{safe edges} and \emph{final
  vertices} (in order to make the reduction more readable). Notice
that this problem amounts to determining whether \runner has a
strategy to avoid crossing an edge with non-zero budget, therefore we
refer to this problem as the \emph{extended safety problem} ($\ESPr$).
Our next step is to encode safe edges and final vertices into
(non-extended) \QSG{s} with gadgets of polynomial size, therefore
proving that the \emph{safety problem} ($\SPr$) is itself \EXP-hard:
$\SPr$ is a special case of the threshold problem $\ThPr_{\supfun}(0)$
with $\supfun$ cost function and threshold~0, for empty initial
distributions. Reductions to threshold problems with other cost
functions close our discussion to prove their \EXP-hardness.

\begin{figure}[tbp]
 \centering
 \resizebox{0.7\textwidth}{!}{%
 \begin{tikzpicture}[xscale=3]
 \node (ABF) at (-1,0) {ABF};
  \node (ESSG) at (0,0) {$\ESPr$};
  \node (SSG) at (1,0) {$\SPr$};
 \node (Sup0) at (0,1) {$\ThPr_{\supfun}(0)$};
 \node (Sup) at (0,2) {$\ThPr_{\supfun}$};
 \node (Disc0) at (-1,1) {$\ThPr_{\discfun{\lambda}}(0)$};
 \node (Disc) at (-1,2) {$\ThPr_{\discfun{\lambda}}$};
 \node (Limsup0) at (1,1) {$\ThPr_{\lsupfun}(0)$};
 \node (Limsup) at (1,2) {$\ThPr_{\lsupfun}$};
 \node (MP0) at (2,1) {$\ThPr_{\MP}(0)$};
 \node (MP) at (2,2) {$\ThPr_{\MP}$};
  
 \path
 (ABF) edge node[el,swap]{Lem.~\ref{lem:abf-to-extended}} (ESSG)
 (ESSG) edge node[el,swap]{Lem.~\ref{lem:safe-extended-wlog}} (SSG)
  (SSG) edge (Sup0)
  (Sup0) edge (Sup)
  (Disc0) edge (Disc)
  (Limsup0) edge (Limsup)
  (MP0) edge (MP)
  (Sup0) edge node[el,swap]{Lem.~\ref{lem:sup-to-ds}} (Disc0)
  (SSG) edge node[el,swap,pos=0.5]{Lem.~\ref{lem:safe-to-limsup}} (Limsup0)
  (Limsup0) edge node[el]{Lem.~\ref{lem:limsup-to-mp}} (MP0)
  ;
  \end{tikzpicture}
  }
  \caption{Reductions used in this section. We denote by $\ThPr_{f}$
    (respectively, $\ThPr_{f}(0)$) the threshold problem
    (respectively, the sub-problem of the threshold problem where
    threshold is $0$) for \QSGs with cost function
    $f$. Non-trivial reductions are labelled with the corresponding
    lemma stated in this section.}
  \label{fig:reductions-dir}
\end{figure}
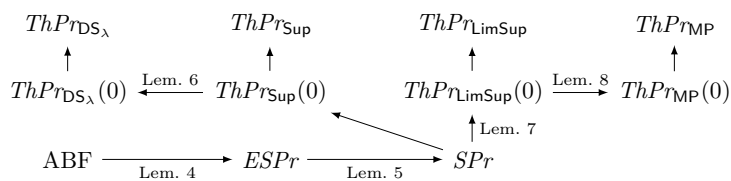

\subparagraph{Alternating Boolean Formula.}  We first recall the
alternating Boolean formula problem (ABF) introduced as game $G_6$
in~\cite{SC79}, which is the \EXP-hard problem from which we perform
our reductions. Intuitively, an ABF is an (infinite) game played on a
Boolean formula whose variables are partitioned into two sets.  Each
player controls the values of one of the sets of variables.  Players
take turns changing the value of one of the variables they control.
The objective of the first player (Prover) is to eventually make the
formula true, while the second player (Disprover) tries to avoid
this. We note that this game closely resembles an infinite horizon
version of the more classical \textsc{QBF Problem}.

More formally, an ABF instance is given by two finite disjoint sets of
Boolean variables, $X$ and $Y$, and a CNF formula over $X\cup Y$.
The game is played by two players called Prover and Disprover. They
take turns changing the value of at most one of the variables they own
($X$ are the variables of Prover, and $Y$ those of Disprover). Prover
wins if and only if the formula is eventually true. A configuration of
this game is thus a pair $(\valuation, \nplayer)$ where $\valuation$
is the current valuation of the variables and $\nplayer$ indicates
which player should play next. The \textsc{ABF problem} consists in,
given an ABF game and an initial configuration, determining whether
Disprover has a winning strategy from the initial configuration. It is
shown \EXP-complete in~\cite{SC79}.

\begin{example}\label{exa:abf}
Consider the formula
$\Phi = \clause_1 \wedge \clause_2 \wedge \clause_3 \wedge
\clause_4$
where $\clause_1 = A \vee \neg C$, $\clause_2 = C\vee D$,
$\clause_3 = C\vee \neg D$ and $\clause_4=B\vee \neg B$. Let
us further consider the partition of the variables into the sets
$X=\{A,B\}$ of Prover, and $Y=\{C,D\}$ of Disprover; and the initial
configuration $(\valuation, \mathrm{Prover})$, where
$\valuation=\{B,C,D\}$ (we denote a valuation by the set of all
variables it sets to true). Clearly, in this initial configuration,
$\Phi$ is false since $\clause_1$ is false. From that
configuration, Prover can either set $A$ to true, or $B$ to false. In
the former case, one obtains the configuration
$(\{A,B,C,D\},\mathrm{Disprover})$, where Prover wins, as $\Phi$ now
evaluates to true. In the latter case, one obtains the configuration
$(\{C,D\},\mathrm{Disprover})$. We claim that, from this
configuration, Prover cannot win the game anymore, i.e., Disprover has
a winning strategy that consists in first setting $C$ to false, and
in, all subsequent rounds, always flipping the value of $D$, whatever
Prover does. Playing according to this strategy ensures Disprover to
force visiting only configurations where either $\clause_2$ or
$\clause_3$ is false.
\end{example}

\subparagraph{Extended \QSG.} To make the encoding of ABF instances
into \QSG easier, we introduce \emph{extended quantitative sabotage
  games} (with $\supfun$ cost function). Those games are \QSG with
$\supfun$ cost function, a designated subset $F \subseteq V$ of
\emph{final vertices} and a designated subset $S\subseteq E$ of
\emph{safe edges} (those special vertices and edges are henceforth
depicted with double lines). $F$ and $S$ influence the semantics of
the game: \saboteur can place some budget on final vertices (which is
accounted for in the cost when \runner visits those vertices), but
cannot put budget on safe edges; and the game stops as soon as \runner
visits a final vertex.  We consider the \emph{extended safety problem}
($\ESPr$), which is to determine whether an extended \QSG $\game$
\emph{with empty initial distribution} has value $\Val(\game)\leq 0$.

Since the cost function is $\supfun$, this amounts to checking that
\runner has a strategy to reach a final vertex, with no budget assigned to it,
without crossing any edge with non-null budget. From now on, we assume
$B < |E|$, as the problem is trivial otherwise. Then:

\begin{lemma}\label{lem:abf-to-extended}
  The ABF problem is polynomial-time reducible to $\ESPr$.
\end{lemma}
\begin{proof}[Sketch]
  We consider an instance of the ABF problem given by Boolean variable
  sets $X$ and $Y$ (owned by Prover and Disprover, respectively) and a
  CNF formula $\Phi$ over $X \cup Y$. We construct an extended \QSG
  $\egame$ such that \saboteur wins in $\egame$ if and only if Prover
  wins in the ABF problem.
  Valuations of the variables in $X \cup Y$ are encoded by budget
  distributions in $\egame$. For each variable $x \in X \cup Y$,
  $\egame$ has $4$ \emph{final} vertices associated with $x$,
  $\Ver(x) = \{\neg x^{(1)}, \neg x^{(2)}, x^{(1)}, x^{(2)} \}$.  A
  budget distribution $\distr$ encodes a valuation in which variable
  $x\in X\cup Y$ is \textbf{true} if and only if
  $\distr(x^{(1)})=\distr( x^{(2)})=1$ and
  $\distr(\neg x^{(1)})=\distr( \neg x^{(2)})=0$.

  Then, $\egame$ simulates the ABF game as follows. The duty of
  \saboteur is to move the budget distribution in such a way that he
  respects the encoding of the valuations explained above. To enforce
  this, we rely on the two gadgets, depicted in
  \figurename~\ref{GadgetGeq2} and~\ref{Check(x)}. They allow \runner
  to check that \saboteur respects the encoding and let him lose if he
  does not. More precisely, the gadget in \figurename~\ref{GadgetGeq2}
  allows one to check that $(i)$ there is a non-zero budget on at
  least two vertices from $\Ver(x)$; and the one in
  \figurename~\ref{Check(x)} that $(ii)$ there is a non-zero budget on
  exactly $\{\neg x^{(1)}, \neg x^{(2)}\}$ or $\{x^{(1)}, x^{(2)}\}$.
  To allow \runner to check one of these conditions, we allow him to
  move to one of the four corner vertices of the corresponding gadget,
  from where one can easily check \runner can win if and only if the
  condition is not respected. In our reduction, \runner will be
  allowed to check condition $(i)$, for all variables, from all
  vertices but will be able to check $(ii)$ only on some of them, as
  we will see later.

  \begin{figure}[tbp]
    \begin{minipage}[b]{0.45\linewidth}
      \centering
      \resizebox{0.9\textwidth}{!}{%
        \begin{tikzpicture}[xscale=2]

          \node[squircle, double] (nX1) at (0,0) {$\neg x^{(1)}$};
          \node[squircle, double] (nX2) at (1,0) {$\neg x^{(2)}$};
          \node[squircle, double] (X1) at (2,0) {$x^{(1)}$};
          \node[squircle, double] (X2) at (3,0) {$x^{(2)}$};

          \node[squircle] (123-1) at (0.5,2) {$\{\neg x^{(1)},\neg x^{(2)},x^{(1)}\}^{(1)}$};
          \node[squircle] (234-1) at (2.5,2) {$\{\neg x^{(2)},x^{(1)},x^{(2)}\}^{(1)}$};
          \node[squircle] (123-2) at (0.5,1) {$\{\neg x^{(1)},\neg x^{(2)},x^{(1)}\}^{(2)}$};
          \node[squircle] (234-2) at (2.5,1) {$\{\neg x^{(2)},x^{(1)},x^{(2)}\}^{(2)}$};

          \draw[-latex,double] (123-1) -- (123-2);
          \draw[-latex,double] (123-2) -- (nX1);
          \draw[-latex,double] (123-2) -- (nX2);
          \draw[-latex,double] (123-2) -- (X1);

          \draw[-latex,double] (234-1) -- (234-2);
          \draw[-latex,double] (234-2) -- (nX2);
          \draw[-latex,double] (234-2) -- (X1);
          \draw[-latex,double] (234-2) -- (X2);

          \node[squircle] (124-1) at (0.5,-2) {$\{\neg x^{(1)},\neg x^{(2)},x^{(2)}\}^{(1)}$};
          \node[squircle] (134-1) at (2.5,-2) {$\{\neg x^{(1)},x^{(1)},x^{(2)}\}^{(1)}$};
          \node[squircle] (124-2) at (0.5,-1) {$\{\neg x^{(1)},\neg x^{(2)},x^{(2)}\}^{(2)}$};
          \node[squircle] (134-2) at (2.5,-1) {$\{\neg x^{(1)},x^{(1)},x^{(2)}\}^{(2)}$};

          \draw[-latex,double] (124-1) -- (124-2);
          \draw[-latex,double] (124-2) -- (nX1);
          \draw[-latex,double] (124-2) -- (nX2);
          \draw[-latex,double] (124-2) -- (X2);

          \draw[-latex,double] (134-1) -- (134-2);
          \draw[-latex,double] (134-2) -- (nX1);
          \draw[-latex,double] (134-2) -- (X1);
          \draw[-latex,double] (134-2) -- (X2);

        \end{tikzpicture}
      }
      \caption{Verifying condition $(i)$}
      \label{GadgetGeq2}
    \end{minipage}
    \hfill
    \begin{minipage}[b]{0.45\linewidth}
      \centering
      \resizebox{0.75\textwidth}{!}{%
        \begin{tikzpicture}[xscale=2]
          \node[squircle, double] (1) at (0,0) {$\neg x^{(1)}$};
          \node[squircle, double] (2) at (1,0) {$\neg x^{(2)}$};
          \node[squircle, double] (3) at (2,0) {$x^{(1)}$};
          \node[squircle, double] (4) at (3,0) {$x^{(2)}$};

          \node[squircle] (13) at (0.5,1) {$\{\neg x^{(1)},x^{(1)}\}$};
          \node[squircle] (14) at (2.5,1) {$\{\neg x^{(1)},x^{(2)}\}$};

          \node[squircle] (23) at (0.5,-1) {$\{\neg x^{(2)},x^{(1)}\}$};
          \node[squircle] (24) at (2.5,-1) {$\{\neg x^{(2)},x^{(2)}\}$};

          \draw[-latex,double] (13) -- (1);
          \draw[-latex,double] (13) -- (3);
          \draw[-latex,double] (14) -- (1);
          \draw[-latex,double] (14) -- (4);

          \draw[-latex,double] (23) -- (2);
          \draw[-latex,double] (23) -- (3);
          \draw[-latex,double] (24) -- (2);
          \draw[-latex,double] (24) -- (4);
        \end{tikzpicture}
      }
      \caption{Verifying condition $(ii)$}
      \label{Check(x)}
    \end{minipage}
  \end{figure}

  The remaining of the construction is done in a way to allow
  \saboteur and \runner to choose valid re-configurations of $\Ver(x)$
  for all variables $x$, and make sure that if a player cheats, it
  allows the other player to win the safety game. If at some point,
  the formula $\Phi$ becomes true, then we allow \saboteur to enter a
  final gadget 
  which verifies that the current budget distribution to
  $\Ver(X) = \bigcup_{x \in X \cup Y} \Ver(x)$ satisfies $\Phi$. This
  last gadget lets \runner choose a clause and then allows \saboteur
  to choose a literal, within this clause, which should be true. It is
  easy to see that the choice of clause $\clause$ can be done by way
  of safe edges. The choice of literal, done by \saboteur, consists in
  choosing a suffix of $\clause$ for which the left-most literal
  holds.
  Figure~\ref{fig:exa-abfMain} shows the $\ESPr$ which results from
  applying our construction to the ABF formula from
  Example~\ref{exa:abf}.  We refer the reader to
  Appendix~\ref{sec:ABF2ESPr} for the full reduction, in particular
  how we can force, before the beginning of the actual game, to start
  in the initial valuation of the ABF game.
\end{proof}

\begin{figure}[tbp]
\begin{center}
\begin{tikzpicture}[scale = 0.8, transform shape]
\tikzset{every vertex/.style={inner sep=2pt}}

\node[squircle, double] (nA1) at (0.59,0) {$\neg A^{(1)}$};
\node[squircle, double] (nA2)	at (1.76,0) {$\neg A^{(2)}$};
\node[squircle, double] (A1) at (2.81,0)  {$A^{(1)}$};
\node[squircle, double] (A2) at (3.74,0)  {$A^{(2)}$};

\node[squircle, double] (nB1) at (4.79,0) {$\neg B^{(1)}$};
\node[squircle, double] (nB2) at (5.97,0) {$\neg B^{(2)}$};
\node[squircle, double] (B1) at (7.03,0)  {$B^{(1)}$};
\node[squircle, double] (B2) at (7.97,0)  {$B^{(2)}$};

\node[squircle, double] (nC1) at (9.03,0) {$\neg C^{(1)}$};
\node[squircle, double] (nC2) at (10.19,0) {$\neg C^{(2)}$};
\node[squircle, double] (C1) at (11.24,0)  {$C^{(1)}$};
\node[squircle, double] (C2) at (12.18,0)  {$C^{(2)}$};

\node[squircle, double] (nD1) at (13.24,0) {$\neg D^{(1)}$};
\node[squircle, double] (nD2) at (14.43,0) {$\neg D^{(2)}$};
\node[squircle, double] (D1) at (15.51,0)  {$D^{(1)}$};
\node[squircle, double] (D2) at (16.47,0)  {$D^{(2)}$};

\node[squircle] (verif) at (15,5) {$\textit{Verif}$}; 
\node[right,scale=0.7] at (verif.south east) {$ABCD\alpha$};
\node[squircle] (play) at (5.5,5) {$Play$}; 
\node[right,scale=0.7] at (play.south east) {$ABCD\alpha$};
\node[squircle] (choose) at (11,5) {$Choose$}; 
\node[right,scale=0.7] at (choose.south east) {$ABCD$};
\node[squircle] (set2) at (11,4) {$set^{(2)}$}; 
\node[right,scale=0.7] at (set2.south east) {$ABCD\alpha$};
\node[squircle] (set1) at (10,3) {$set^{(1)}$}; 
\node[right,scale=0.7] at (set1.south east) {$AB\alpha$};

\node[squircle] (setnA1) at (1.5,3) {$set^{(1)}_{\neg A}$}; 
\node[right,scale=0.7] at (setnA1.south east) {$BCD\alpha$};
\node[squircle] (setnA2) at (1.5,2) {$set^{(2)}_{\neg A}$};
\node[right,scale=0.7] at (setnA2.south east) {$ABCD\alpha$};
\node[squircle] (setA1) at (3.5,3) {$set^{(1)}_{A}$}; 
\node[right,scale=0.7] at (setA1.south east) {$BCD\alpha$};
\node[squircle] (setA2) at (3.5,2) {$set^{(2)}_{A}$};
\node[right,scale=0.7] at (setA2.south east) {$ABCD\alpha$};

\node[squircle] (setnB1) at (5.5,3) {$set^{(1)}_{\neg B}$}; 
\node[right,scale=0.7] at (setnB1.south east) {$ACD\alpha$};
\node[squircle] (setnB2) at (5.5,2) {$set^{(2)}_{\neg B}$};
\node[right,scale=0.7] at (setnB2.south east) {$ABCD\alpha$};
\node[squircle] (setB1) at (7.5,3) {$set^{(1)}_{B}$}; 
\node[right,scale=0.7] at (setB1.south east) {$ACD\alpha$};
\node[squircle] (setB2) at (7.5,2) {$set^{(2)}_{B}$};
\node[right,scale=0.7] at (setB2.south east) {$ABCD\alpha$};

\node[squircle] (Cl1) at (3,-2) {$\clause_1$};
\node[right,scale=0.7] at (Cl1.south east) {$ABCD$};
\node[squircle] (Cl2) at (7,-2) {$\clause_2$};
\node[right,scale=0.7] at (Cl2.south east) {$ABCD$};
\node[squircle] (Cl3) at (11,-2) {$\clause_3$};
\node[right,scale=0.7] at (Cl3.south east) {$ABCD$};
\node[squircle] (Cl4) at (15,-2) {$\clause_4$};
\node[right,scale=0.7] at (Cl4.south east) {$ABCD$};

\node[squircle, double] (alpha) at (3.5,5) {$\ \alpha\ $}; 

\draw[-latex] (choose) -- (play);
\draw[-latex] (choose) -- (verif);

\draw[-latex, double] (set2) -- (choose);
\draw[-latex, double] (set1) -- (set2);

\draw[-latex, double] (play) -- (setnA1);
\draw[-latex, double] (play) -- (setA1);
\draw[-latex, double] (play) -- (setnB1);
\draw[-latex, double] (play) -- (setB1);
\draw[-latex, double] (setnA1) -- (setnA2);
\draw[-latex, double] (setA1) -- (setA2);
\draw[-latex, double] (setnB1) -- (setnB2);
\draw[-latex, double] (setB1) -- (setB2);
\draw[-latex, double] (setnA2) -- (nA1);
\draw[-latex, double] (setnA2) -- (nA2);
\draw[-latex, double] (setA2) -- (A1);
\draw[-latex, double] (setA2) -- (A2);
\draw[-latex, double] (setnB2) -- (nB1);
\draw[-latex, double] (setnB2) -- (nB2);
\draw[-latex, double] (setB2) -- (B1);
\draw[-latex, double] (setB2) -- (B2);

\draw[-latex,rounded corners = 2pt, double] (setnA2) -- (1.5,1) -- (10,1) -- (set1);
\draw[-latex,rounded corners = 2pt, double] (setA2) -- (3.5,1) -- (10,1) -- (set1);
\draw[-latex,rounded corners = 2pt, double] (setnB2) -- (5.5,1) -- (10,1) -- (set1);
\draw[-latex,rounded corners = 2pt, double] (setB2) -- (7.5,1) -- (10,1) -- (set1);

\draw[-latex,rounded corners = 2pt, double] (verif) -- (15,3) -- (17,3) -- (17,-2.7) -- (3,-2.7) -- (Cl1);
\draw[-latex,rounded corners = 2pt, double] (verif) -- (15,3) -- (17,3) -- (17,-2.7) -- (7,-2.7) -- (Cl2);
\draw[-latex,rounded corners = 2pt, double] (verif) -- (15,3) -- (17,3) -- (17,-2.7) -- (11,-2.7) -- (Cl3);
\draw[-latex,rounded corners = 2pt, double] (verif) -- (15,3) -- (17,3) -- (17,-2.7) -- (15,-2.7) -- (Cl4);

\draw[-latex, double] (Cl1) -- (A1.south);
\draw[-latex, double] (Cl1) -- (nC1.south);

\draw[-latex, double] (Cl2) -- (C1.south);
\draw[-latex, double] (Cl2) -- (D1.south);

\draw[-latex, double] (Cl3) -- (C1.south);
\draw[-latex, double] (Cl3) -- (nD1.south);

\draw[-latex, double] (Cl4) -- (B1.south);
\draw[-latex, double] (Cl4) -- (nB1.south);

\end{tikzpicture}
\end{center}
\caption{Excerpt of the $\ESPr$ constructed from the ABF of
  Example~\ref{exa:abf}. In addition to these nodes and edges, the
  full $\ESPr$ contains: an initialisation gadget; a \emph{safe} edge
  from a node $n$ to all four corner nodes of gadget $(i)$ in
  \figurename~\ref{GadgetGeq2} iff $n$ is labeled by $\alpha$; and a
  \emph{safe} edge from a node $n$ to all four corner nodes of gadget
  $(ii)$ in \figurename~\ref{Check(x)} testing variable
  $x\in\{A,B,C,D\}$ iff $n$ is labeled by~$x$. These parts have been
  omitted for the sake of clarity.}
\label{fig:exa-abfMain}
\end{figure}
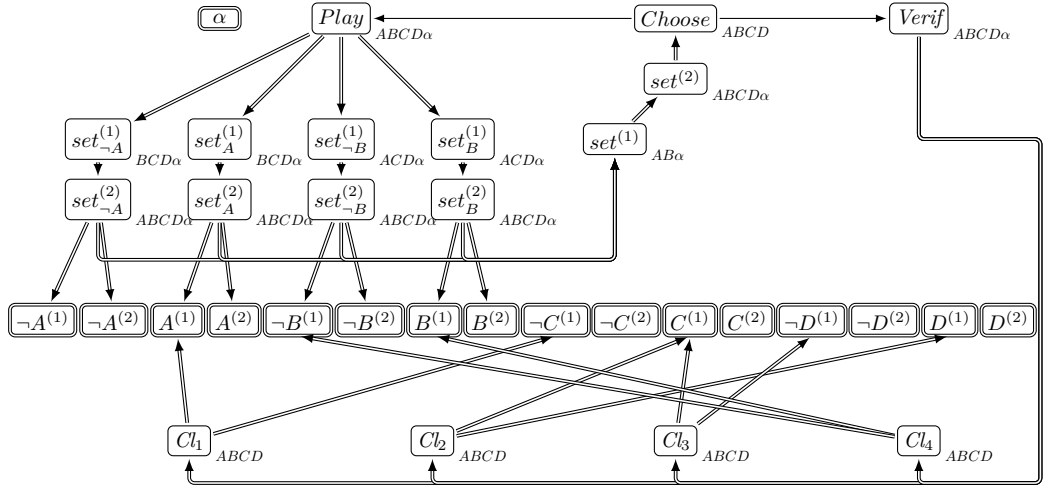

We now explain how to encode safe edges and final vertices into usual
\QSG{s}, therefore showing the \EXP-hardness of the safety problem for
\QSG{s}.

\begin{lemma}\label{lem:safe-extended-wlog}
  The extended safety problem $\ESPr$ is polynomial-time reducible to a safety
  problem $\SPr$ with budget $2$.
\end{lemma}
\begin{proof}[Sketch]
  Each final vertex $v$ in an extended \QSG $\egame$ is replaced by
  the gadget in \figurename~\ref{finalNodes}, where
  $\{\alpha_i \st 1 \le i \le B + 1\}$ is a clique of size $B+1$,
  hence bigger than the budget of \saboteur. To encode $\distr(v)=1$
  in $\egame$, \saboteur now puts one unit of budget on
  $(\nd{A},\nd{C_1})$. If \runner reaches the gadget (through
  $\nd{A}$), \saboteur puts one unit of budget on
  $(\nd{A},\nd{C_2})$. Clearly, \runner loses if and only if there was
  already one unit on $(\nd{A},\nd{C_1})$ (i.e., $v$ was marked in
  \egame).
  Each safe edge $(\nd{A},\nd{C})$ is replaced by the gadget in
  \figurename~\ref{safeEdges}. Here, we make use of final vertices and
  disjoint paths so that \saboteur cannot block all paths from
  $\nd{A}$ to $\nd{C}$ without letting \runner win by visiting a final
  vertex with zero budget. Both gadgets have polynomial size since we
  assume that $B<|E|$. 
\end{proof}

   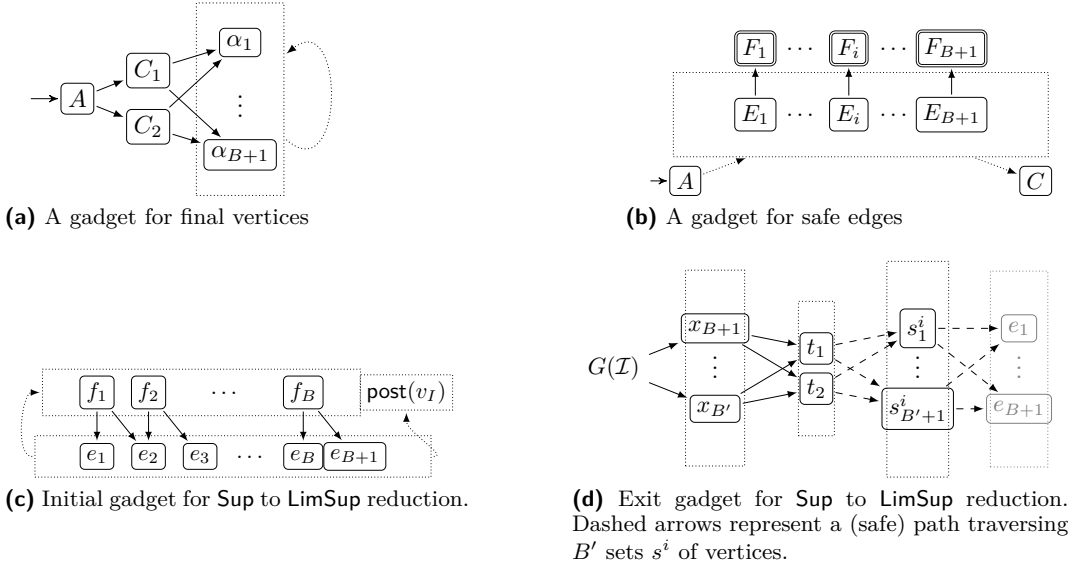
\begin{figure}
    \subfloat[A gadget for final vertices]{
    \resizebox{0.32\textwidth}{!}{%
	    \begin{tikzpicture}[xscale=1,yscale=0.8]

	      \node[squircle,initial] (A) at (0,1) {$A$};
	      \node[squircle] (C1) at (1,0.5) {$C_2$};
	      \node[squircle] (C2) at (1,1.5) {$C_1$};

	      \node[squircle] (alph1) at (2.3,2) {$\alpha_1$};
	      \node (alphdots) at (2.3,1) {$\vdots$};
	      \node[squircle] (alph2) at (2.3,0) {$\alpha_{B+1}$};

	      \node[fit=(alph1) (alph2)] (alphgroup) {};

	      \path
	      (A) edge (C1)
	      (A) edge (C2)
	      (C1) edge (alph1)
	      (C1) edge (alph2)
	      (C2) edge (alph1)
	      (C2) edge (alph2)
	      (alphgroup) edge[loop,looseness=2,out=-45, in=45,densely dotted]
	      	(alphgroup)
	      ;

      \end{tikzpicture}
    }
    \label{finalNodes}
    }
    \hfill
    \subfloat[A gadget for safe edges]{
    \resizebox{0.4\textwidth}{!}{%
	    \begin{tikzpicture}[xscale=0.7,yscale=0.5]
	    \node[squircle,initial] (A) at (-0.5,-2) {$A$};
	    
	    \node[squircle] (E1) at (1,0) {$E_1$};
	    \node (edots1) at (2,0) {\dots};
	    \node[squircle] (Ei) at (3,0) {$E_i$};
	    \node (edots2) at (4,0) {\dots};
	    \node[squircle] (EB1) at (5.2,0) {$E_{B+1}$};

	    \node[squircle] (C) at (7,-2) {$C$};

	    \node[squircle,double] (F1) at (1,2) {$F_1$};
	    \node (fdots1) at (2,2) {\dots};
	    \node[squircle,double] (Fi) at (3,2) {$F_i$};
	    \node (fdots2) at (4,2) {\dots};
	    \node[squircle,double] (FB1) at (5.2,2) {$F_{B+1}$};

	    \node[fit=(E1) (EB1)] (egroup) {};

	    \path
	    (A) edge[densely dotted] (egroup)
	    (egroup) edge[densely dotted] (C)
	    (E1) edge (F1)
	    (Ei) edge (Fi)
	    (EB1) edge (FB1)
	    ;
      \end{tikzpicture}
    }
    \label{safeEdges}
    }
    
    \subfloat[Initial gadget for $\supfun$ to $\lsupfun$ reduction.]{
      \resizebox{0.42\textwidth}{!}{%
        \begin{tikzpicture}[xscale=0.8,yscale=1]
		\node[rectangle,draw,densely dotted] (vi) at (6,1)
		{$\mathsf{post}(v_I)$};

          \node[squircle] (e1) at (0,0) {$e_1$};
	  \node[squircle] (e2) at (1,0) {$e_2$};
	  \node[squircle] (e3) at (2,0) {$e_3$};
	  \node (edots) at (3,0) {\dots};
	  \node[squircle] (eB) at (4,0) {$e_B$};
	  \node[squircle] (eB1) at (5,0) {$e_{B+1}$};

          \node[squircle] (f1) at (0,1) {$f_1$};
          \node[squircle] (f2) at (1,1) {$f_2$};
	  \node (fdots) at (2.5,1) {\dots};
	  \node[squircle] (fB) at (4,1) {$f_B$};

	  \node[fit=(e1) (eB1),draw,densely dotted,rectangle] (egroup) {};
	  \node[fit=(f1) (fB),draw,densely dotted,rectangle] (fgroup) {};

	  \path
	  (f1) edge (e1)
	  (f1) edge (e2)
	  (f2) edge (e2)
	  (f2) edge (e3)
	  (fB) edge (eB)
	  (fB) edge (eB1)
	  (egroup) edge[in=180,out=180,densely dotted] (fgroup)
	  (egroup) edge[out=0,in=-90,densely dotted] (vi)
	  ;
        \end{tikzpicture}
      }
      \label{fig:initial-gadget}
      }
    \hfill
    \subfloat[Exit gadget for $\supfun$ to $\lsupfun$ reduction. Dashed arrows
    represent a (safe) path traversing $B'$ sets $s^i$ of vertices.]{
      \resizebox{0.45\textwidth}{!}{%
        \begin{tikzpicture}[xscale=1.5,yscale=0.6]
	  \node (GI) at (0,1) {$G(\mathcal{I})$};

	  \node[squircle] (xB1) at (1,2) {$x_{B+1}$};
	  \node (xdots) at (1,1.2) {$\vdots$};
	  \node[squircle] (xBE) at (1,0) {$x_{B'}$};

	  \node[squircle] (t1) at (2,1.5) {$t_{1}$};
	  \node[squircle] (t2) at (2,0.5) {$t_{2}$};

	  \node[squircle] (s1i) at (3,2) {$s_1^i$};
	  \node (xdots) at (3,1.2) {$\vdots$};
	  \node[squircle] (sBE1i) at (3,0) {$s_{B'+1}^i$};

	  \node[squircle,gray] (e1) at (4,2) {$e_{1}$};
	  \node[gray] (edots) at (4,1.2) {$\vdots$};
	  \node[squircle,gray] (eB1) at (4,0) {$e_{B+1}$};

	  \node[fit=(xB1) (xBE),draw,densely dotted,rectangle] (xgroup) {};
	  \node[fit=(t1) (t2),draw,densely dotted,rectangle] (tgroup) {};
	  \node[fit=(e1) (eB1),draw=gray,densely dotted,rectangle] (egroup) {};
	  \node[fit=(s1i) (sBE1i),draw,densely dotted,rectangle] (safegroup) {};

	  \path
	  (GI) edge (xB1)
	  (GI) edge (xBE)
	  (xB1) edge (t1)
	  (xB1) edge (t2)
	  (xBE) edge (t1)
	  (xBE) edge (t2)
	  (t1) edge[dashed] (s1i)
	  (t1) edge[dashed] (sBE1i)
	  (t2) edge[dashed] (s1i)
	  (t2) edge[dashed] (sBE1i)
	  (s1i) edge[dashed] (e1)
	  (s1i) edge[dashed] (eB1)
	  (sBE1i) edge[dashed] (e1)
	  (sBE1i) edge[dashed] (eB1)
	  ;
        \end{tikzpicture}
      }
      \label{fig:exit-gadget}
      }
    \caption{Dotted arrows represent edges from all sources to all targets.}
  \end{figure}

As the safety problem is a specific case of the threshold problem for
$\supfun$ \QSGs (where the initial distribution is empty, and
threshold is fixed to $0$), it follows that $\ThPr_{\supfun}(0)$ and
$\ThPr_{\supfun}$ are \EXP-hard too. 

We note that given a \QSG $\game$, for all plays $\pi$ in $\game$, for
all $0 < \lambda < 1$, and for all $\delta \in \Delta(\game)$,
$\supfun(\overline{\pi}) = 0$ if and only if
$\discfun{\lambda}(\overline{\pi}) = 0$. This implies the following
result, showing that $\ThPr_{\discfun{\lambda}}(0)$ and
$\ThPr_{\discfun{\lambda}}$ are also \EXP-hard.
\begin{lemma}\label{lem:sup-to-ds}
  For any $\lambda \in (0,1)$, the threshold problem for
  $\discfun{\lambda}$ and threshold $0$ is equivalent to the threshold
  problem for $\supfun$ and threshold $0$.
\end{lemma}

Let us now focus on $\lsupfun$. To show that $\ThPr_{\lsupfun}$ is
\EXP-hard, we describe a reduction from $\SPr$ to
$\ThPr_{\lsupfun}(0)$ as stated in the following lemma.

\begin{lemma}\label{lem:safe-to-limsup}
  The safety problem $\SPr$ is polynomial-time reducible to the threshold
  problem for $\lsupfun$ and threshold $0$.
\end{lemma}
\begin{proof}[Sketch]
  Let $\mathcal{I} = (V, E, B, v_I, \distr_I, \supfun)$ be an instance of $\SPr$
  (with $G(\mathcal{I})$ its underlying graph $(V,E)$).  We build a \QSG $\game$
  with cost function $\lsupfun$ such that $\Val(\game)=0$ if and only if \runner
  wins in $\mathcal{I}$. The idea of the construction is that a play of $\game$
  consists in simulating a potentially infinite sequence of plays of
  $\mathcal{I}$, using appropriate gadgets to `reset' the safety game between
  two successive simulations. Then, repeatedly playing a winning strategy for
  $\mathcal{I}$ allows \runner to ensure a $\lsupfun$ of $0$ in $\game$; and one
  can extract a winning strategy for the safety game $\mathcal{I}$ from any
  strategy ensuring a $\lsupfun$ of $0$ in $\game$.  The \QSG $\game$ has budget
  $B' = |E|$ and is obtained by extending $G(\mathcal{I})$ with two gadgets.
  Note that we are giving \saboteur more budget than he had in $\mathcal{I}$.
  However, as we will see in the sequel, at the beginning of every faithful
  simulation of $\mathcal{I}$ (i.e. when \runner moves to $G(\mathcal{I})$)
  there will be $B' - B$ of it in the second gadget and $B$ in the first and
  during any faithful simulation of $\mathcal{I}$ only budget from the
  initial gadget is redistribtued into $G(\mathcal{I})$.

  The first gadget is an initial gadget which is visited every time the safety game is
  `reset'. It allows \runner to stay safe from any weighted edges (and
  avoid reaching $G(\mathcal{I})$) until \saboteur has placed $B$ units of
  budget on it (and thus removed them from the $G(\mathcal{I}$). It is depicted
  in \figurename~\ref{fig:initial-gadget}, where all $e_i$ are intuitively
  copies of $v_I$, and $\mathsf{post}(v_I)$ corresponds to the set of all successors of $v_I$ in
  $G(\mathcal{I})$.

  The second gadget allows \runner to leave $G(\mathcal{I})$ if
  \saboteur ever places more than $B$ units of budget on
  $G(\mathcal{I})$ (and thus removes this budget from the gadgets),
  thereby triggering a `reset' of the simulation. This gadget,
  depicted in \figurename~\ref{fig:exit-gadget}, also allows
  \runner to come back to the initial gadget visiting only edges with
  zero budget. The figure shows a sequence of safe transitions (i.e. several
  vertices with high out-degree) which leads back to the copies $e_i$ of the
  initial vertex. Further, this `safe path' takes long enough for
  \saboteur to redistribute the budget from $G(\mathcal{I})$ to both
  gadgets. In order for \saboteur to stop \runner from always taking this
  `safe exit' from $G(\mathcal{I})$ he can place $B' - B$ budget in specific
  edges of this second gadget. More specifically, he can place a unit of budget
  on one outgoing edge from each $x_j$, for $B+1  \le j \le B'$, before forcing
  \runner to enter $G(\mathcal{I})$.

  \item \subparagraph{Intuition behind the global construction.}
  Assume that \saboteur has a winning strategy in $\mathcal{I}$. Then, when
  \runner is in the initial gadget, \saboteur will play as expected and remove
  all weights from $G(\mathcal{I})$. Critically, the weights he removes from 
  $G(\mathcal{I})$ will go to specific edges in both gadgets described above.
  \runner is now forced to play into $G(\mathcal{I})$, and \saboteur can follow
  his winning strategy to hit \runner at some point without using more than $B$
  weights. If \runner attempts to bail out of $G$ through the alternative exit,
  and to head back to the initial gadget, then we make sure he is also hit by
  \saboteur. Clearly, this ensures that the $\lsupfun$ value of the game is
  strictly greater than $0$.  Now assume that \runner has a winning strategy in
  $\mathcal{I}$. In this case, if \saboteur does not remove all weights from
  $G(\mathcal{I})$, then \runner is allowed to stay in the initial gadget
  forever or jump to $G(\mathcal{I})$ and immediately bail out using the exit
  gadget. In both cases he avoids getting hit by \saboteur. Let us assume
  \saboteur plays as expected and thus \runner enters $G(\mathcal{I})$
  eventually.  In this case, \runner can play his winning strategy, hence
  avoiding edges with non-zero budget (with \saboteur using budget $B$). Either
  he dodges weighted edges forever, or \saboteur cheats and uses some of his
  additional budget.  However, in this case he creates an exit for \runner back
  to the initial gadget, and the same analysis as above applies. This implies
  that the value of the game is exactly $0$.
\end{proof}

Proving the \EXP-hardness result for cost function \mpfun is done by
noticing that, for threshold $0$, both problems are equivalent.

\begin{lemma}\label{lem:limsup-to-mp}
  The threshold problem for $\lsupfun$ and threshold $0$ is polynomial-time
  reducible to the threshold problem for $\mpfun$ and threshold $0$.
\end{lemma}


\section{Static quantitative sabotage games}\label{sec:static}

In light of the \EXP-completeness of \QSGs, we
study in this section a restriction of the problem, that might be
sufficient to model some interesting cases.
The restriction concerns the dynamics of the behaviour of \saboteur.
In a \emph{static} \QSG, \saboteur chooses at the beginning a budget
distribution (hence, changing the initial budget distribution), and
then commits to this distribution during the whole game. The situation
is no longer a reactive two-player game, but rather we ask whether for
every possible initial (and static) budget distribution, \runner has a
nicely behaved strategy.

Formally, for a \QSG $\game=(V,E,B,v_I,f)$ (we remove the initial
budget distribution from the tuple in this section, since it is
useless) and a budget distribution $\distr \in \Delta(\game)$, we
denote by $\gameStat{\distr}$ the \QSG obtained from $\game$ by taking
$\distr$ as initial budget distribution. Furthermore, we define the
\emph{identity strategy} $\iota$ of \saboteur in $\game$, as the
strategy mapping every prefix $\pi\in\PrefsSaboteur(\game)$ to the
last budget distribution appearing in prefix $\pi$. We let
$\Valstatic(\game)=\sup_{\distr\in\Delta(\game)}\inf_{\rho \in
  \StratRunner(\game)} f(\pi^\distr_{\rho,\iota})$,
where $\pi^\distr_{\rho,\iota}$ denotes the unique play defined by the
profile $(\rho,\iota)$ in \QSG $\gameStat{\distr}$. Notice that this
value is equal to
$\inf_{\rho \in \StratRunner(\game)}\sup_{\distr\in\Delta(\game)}
f(\pi^\distr_{\rho,\iota})$,
since in $\game$, when \saboteur follows strategy $\iota$, the
quantitative game $\sem\game$ (see Appendix~\ref{sec:qgames-sem}) is
split into independent games, one for each initial distribution
$\distr$, that \runner knows as soon as it starts playing. The
\textsc{Static Threshold problem with cost function $f$} consists in,
given as input a \QSG $\game$ with cost function $f$ and a
non-negative threshold $T$, determining whether the inequality
$\Valstatic(\game) \leq T$ holds.  We now state the complexity of this
new problem.

\begin{theorem}
  For cost functions $\inffun$ and $\linffun$, the static threshold
  problem over \QSG{s} is in \P; for $\supfun$,
  $\lsupfun$, $\mpfun$, and $\dsfun$, it is \coNP-complete.
\end{theorem}

First, we give the intuition behind our polynomial-time algorithm to decide the
static threshold problem for cost functions $\inffun$ and $\linffun$.

\begin{lemma}\label{lem:poly-inf-linf}
	For cost functions $\inffun$ and $\linffun$, the static threshold
	problem over \QSGs is in \P.
\end{lemma}
\begin{proof}[Sketch]
For $\inffun$, we claim that
$\Valstatic(\game)=\lfloor |\overline E|/B \rfloor$, where
$\overline E$ is the set of edges reachable from $v_I$.
Indeed once a distribution $\distr$ is chosen, any optimal strategy of
\runner will make him reach an edge of $\overline E$ that has the minimum
weight, thus \saboteur must distribute evenly its budget over $\overline E$. 
A similar argument works for $\linffun$, showing that
$\Valstatic(\game)=\lfloor |\widetilde E|/B \rfloor$, where $\widetilde E$ is
the set of edges reachable from $v_I$ and contained in a strongly connected
component.
\end{proof}

Then, let us turn to the  $\coNP$-completeness of the problem  for cost
functions $\supfun$, $\lsupfun$, $\mpfun$, and $\dsfun$. Notice that,
because of the two possible definitions of $\Valstatic(\game)$
explained in the beginning of the section, the complement of the
static threshold problem asks whether there exists a budget
distribution $\distr$ such that
$f(\pi^\distr_{\rho,\iota})>T$ for every strategy
$\rho \in \StratRunner(\game)$ of \runner. Thus we  show the
$\NP$-completeness of the complement of the static threshold problems
for the four cost functions.

\begin{lemma}\label{lem:coNP-complete}
	For cost functions $\supfun$, $\lsupfun$, $\mpfun$, and $\dsfun$, the
	complement of the static threshold problem over \QSGs is \NP-complete.
\end{lemma}

\begin{proof}[Sketch]
For the membership in $\NP$, we can first guess a budget distribution
$\distr$ (that is of size polynomial), and then compute the value of
the one-player (since player \Max has no choices anymore) quantitative
game $\gameStat{\distr}$, to check if it is greater than~$T$:
computing the value of such a game 
can be done in polynomial time for the four cost functions we
consider (see~\cite{ag11}).

For the $\NP$-hardness with cost functions $\lsupfun$ and $\mpfun$, we
give a reduction from the following problem. The \textsc{Feedback arc
  set problem} asks, given a directed graph $G = (\vertices,\edges)$
and a threshold $k \leq |\edges|$, whether there is a set $E'$ of at
most $k$ edges of $G$ such that $(\vertices,\edges \setminus E')$ is
acyclic.
Karp showed \cite{Kar72} that the feedback arc set problem is
$\NP$-complete.
Let us consider an instance of the feedback arc set problem, given by
a directed graph $G = (\vertices,\edges)$ and a natural integer
$k\le |\edges|$. Wlog, we can add to the graph a vertex $v_I$, with
null in-degree, and, for all vertices $v\neq v_I$, an edge
$(v_I,v)$. Observe that this does not change the output of the
feedback arc set problem as $v_I$ is not included in any cycle.
We then construct a \QSG $\game = (\vertices, \edges,k,v_I,f)$ with
$f\in\{\lsupfun,\mpfun\}$. It is not difficult to show that
$\Valstatic(\game)>0$ if and only if there exists a set $E'$ of $k$
edges of $G$ such that $(\vertices,\edges\setminus E')$ is
acyclic.
The result for $\supfun$ and $\dsfun$ is then obtained by a slight
modification of the previous proof. In particular, we make use of
Lemma~\ref{lem:sup-to-ds}, once more. We refer the reader to
Appendix~\ref{app:coNP} for the details.
\end{proof}

\section{Reactive systems under failure}

One can see a sabotage game as a system in which a controller tries to
evolve while avoiding as much as possible the failures caused by the
environment. The vertices of the graph represent configurations of the
system, edges represent the actions, and the budget of the Saboteur
may represent a finite amount of failures that can simultaneously
occur during the execution. In a quantitative reasoning, a failure may
be better represented by a quantity describing how much some elements
of the system are overloaded, and then how much it would cost, in
terms of time or energy, to use them.

Following this main motivation, we propose to look at sabotage games
as a particular semantics of controllable systems. Indeed, while a
standard semantics would analyse the feasibility of a requirement in a
fully functional system, a \emph{sabotage semantics} allows one to
analyse systems subject to errors, and to decide, e.g., whether one
can satisfy a Boolean constraint while minimising the average number
of failures encountered during the execution. In particular, sabotage
games, as introduced in this work, would correspond to the sabotage semantics of
a system where the controller must walk in a graph with no particular objective,
other than minimising the failures.

From a modelling point of view, graphs---which can be viewed as
one-player games with trivial winning conditions---are quite
limited. In more realistic models, we may be interested in modelling
systems with uncontrollable actions (i.e., two-player games), and
where the controller has a specific Boolean goal to achieve, instead
of simply staying in the graph \emph{ad vitam \ae{}ternam}. A more
realistic goal is usually expressed via a parity condition or LTL
formulas. In Appendix~\ref{sec:more-expressive}, we show that when a
reactive system is modelled by a two-player parity game, deciding
whether one can ensure the parity condition, while maintaining a cost
associated with the sabotage semantics below a given threshold, is not
harder than solving sabotage games. That is, the problem is
\EXP-complete. This result is obtained by a reduction to quantitative
parity games~\cite{CHJ05}.  When the requirement is expressed with an
LTL formula instead of a parity condition, the problem becomes
2-\EXP-complete, due to an additional exponential blow-up in the size
of the input formula.  Note, however, that the LTL-reactive synthesis
problem itself (with the standard non-sabotage semantics) is already
2-\EXP-complete. In this case, the sabotage semantics does not add to
the complexity of the problem, which further shows that our present
contributions might have practical applications, albeit the high
complexity.

\section{Conclusion}\label{sec:conclusions}

We have conducted a study of systems subject to failure, using the
model of \emph{quantitative sabotage games}. We have shown that under
\emph{dynamic sabotage}, the threshold problem is \EXP-complete for
most objective functions, and \coNP-complete under \emph{static
  sabotage}, for the same functions (see table~\ref{tbl:summary} for a
summary of these results). We have also shown the applicability of our
framework to deal with the more general problem of reactive synthesis
in systems under failures.
The \QSGs we have introduced open many questions related to evolving
structures. Here we have studied the worst-case scenario, i.e.,
where the environment is modelled by an antagonistic adversary, but,
as considered in \cite{KleRad10} for reachability Boolean objectives,
one could also look at a probabilistic model, where failures, i.e.,
redistributions of weights, are random variables. Another natural
extension of this work would be to consider a more realistic setting
where the controller (\runner) has partial information regarding the
weights of \saboteur.

Although the synthesis problem has been widely studied in theory,
there are not many tools which implement the known theoretical
solutions to decide it. The is is particularly true for quantitative
objectives.  Recently, however, competitions have been organised to
encourage the development of such tools and the standardisation of an
input format (see, e.g., SYNTCOMP and SyGuS).\footnote{Links to both
  competitions' websites: \url{http://www.syntcomp.org} and
  \url{http://www.sygus.org/}.}  Motivated by the similarities between
the ABF problem (solving a safety game described by a logical formula)
and the synthesis problem as solved in those competition (solving a
safety game described by a logical circuit), one of our future
projects is to show that quantitative extensions of some of the
practical tools implemented for the reactive synthesis problem could
be used to solve sabotage games.

\clearpage

\bibliographystyle{abbrv}

\newpage


\appendix 

\changepage{3cm}{3cm}{-1.5cm}{-1.5cm}{}{-1cm}{}{}{}

\section{Relation with cops and robbers game}\label{CopsAndRobbers}

We observe that the result on safety games is related to the Cops and
Robbers games studied mostly by the graph theoretical community (see,
e.g.,~\cite{ag11} and references therein for a survey).  We remark
that Cops and Robbers games are usually defined as played on the
vertices of undirected graphs.  In~\cite{gr95} it was shown that
several variants of the Cops and Robbers game without helicopters and,
as usual, played on the vertices of an undirected graph, are
\EXP-complete. In contrast, our result implies that the Cops and
Robbers game played on the edges of a graph with $B$ cops, one
\emph{helicopter} and a \emph{slow robber}, i.e., which can traverse
at most one edge per turn, is \EXP-hard. A similar version is studied
in~\cite{SeyTho93}, where they consider helicopters and a \emph{fast
  robber}.  However, the game is played on the vertices of an
undirected graph and the complexity of solving the game is left open
in that paper.  It is easy to lift our results to games where weights
are placed on vertices and no longer on edges by considering line
graphs: in contrast, the other direction from vertices to edges would
have been more difficult, and is not currently known for the best of
our knowledge.

\section{Encoding of quantitative sabotage games in quantitative
  games}\label{sec:qgames-sem}\label{sab2games}

In this section, we give formal definitions of quantitative two-player
games, and show an exponential encoding of \QSGs into these games.
Thereafter, we choose to call $\Min$ and $\Max$ the two players of our
games, to distinguish them from \runner and \saboteur, used in the
main part of this article.

\subsection{Two-player games} 

A weighted arena is a tuple $G=(V_\Min,V_\Max,E,w,v_I)$ with
$V=V_\Min \uplus V_\Max$ a finite set of vertices partitioned into the
set $V_\Min$ of vertices of player $\Min$ and the set $V_\Max$ of
vertices of player $\Max$, $E\subseteq V^2$ is a set of edges,
$w\colon E \to \N$ is a weight function assigning an integer weight to
each edge of the arena, and $v_I\in V$ is an initial vertex. Given a
weight function $w\colon E \to \R$, we write $|w|$ for the greatest
weight in $w$, i.e., $|w|= \max_{e\in E} w(e)$.

Intuitively, the two players $\Min$ and $\Max$ move a token along the
edges of the graph $(V,E)$, starting on vertex $v_I$. When the token
is on a vertex of $V_\Min$, it is $\Min$ that chooses the next vertex,
and when on $V_\Max$, it is $\Max$. To allow them to play infinitely,
we make the assumption that every vertex $v$ has an outgoing edge,
i.e., that there exists $(v,v')\in E$. A strategy for a player is
simply a mapping telling him what to play depending on the
past. Formally, given an arena $G=(V_\Min,V_\Max,E)$, a play is an
infinite sequence of vertices $v_0v_1v_2\cdots\in V^\omega$ such that
$v_0=v_I$, and $(v_i,v_{i+1})\in E$ for all $i\ge 0$. We say that a
prefix $v_0\cdots v_k$ of a play belongs to \Min (respectively, \Max)
if $v_k\in V_\Min$ (respectively, $v_k\in V_\Max$). A strategy for
player $p$ is a mapping $\sigma$ from prefixes of plays belonging to
$p$ to vertices such that $(v_k,\sigma(v_0 \cdots v_k))\in E$ for all
prefix $v_0 \cdots v_k$ belonging to $p$. The outcomes of a strategy
$\sigma$ of player $p$ are all plays $v_0 v_2 \cdots$ such that for
all $v_0\cdots v_k$ with $v_k\in V_p$,
$v_{k+1}= \sigma (v_0\cdots v_k)$.  We write $\outcomes(G)$ the set of
plays in $G$ (we omit $G$ when it is clear from the context),
$\outcomes(\sigma)$ the set of outcomes of a strategy $\sigma$, and
$\outcomes(\sigma_\Min,\sigma_\Max)$ the only play contained in
$\outcomes(\sigma_\Min) \cap \outcomes (\sigma_\Max)$.

Since we are dealing with quantitative game, we use a \emph{value
  function} to map plays to values in
$\overline \R=\R\uplus\{+\infty\}$. A quantitative game is a pair
$(G,f)$ consisting of an arena $G$ and such a value function $f$. Most
standard value functions are defined by using the weights in the
weighted arena: equipped of one of the cost functions
$f\colon \R^\omega \to \overline\R$ described in the main part of the
article ($\inffun$, $\linffun$, $\supfun$, $\lsupfun$, $\mpfun$, or
$\discfun{\lambda}$ for instance), we may define a value function
$f_w$ by setting $f_w (v_0v_2\cdots) = f(w(v_0,v_1) w(v_1,v_2)\cdots)$
for all plays $v_0v_1\cdots$.

In a quantitative game $(G,f)$, the value of a strategy $\sigma_\Min$
(respectively, $\sigma_\Max$) of $\Min$ (respectively, $\Max$) is:
\[ \Val((G,f),\sigma_\Min) = \sup_{\pi\in \outcomes(\sigma_\Min)}
f(\pi), \quad \Val((G,f),\sigma_\Max) = \inf_{\pi\in
  \outcomes(\sigma_\Max)} f(\pi).\]
To characterise the best value that each player can guarantee no
matter what the opponent is doing, we consider the upper value
$\uppervalue$ (the best $\Min$ can hope for) and lower value
$\lowervalue$ (the best $\Max$ can hope for), defined by:
\[\uppervalue(G,f) = \inf_{\sigma_\Min} \Val ((G,f),\sigma_\Min), \quad
\lowervalue(G,f) = \sup_{\sigma_\Max} \Val ((G,f),\sigma_\Max).\]

\begin{proposition}
  In quantitative games, for all the value functions $f_w$ obtained by
  considering the cost functions $f$ used above, upper and lower
  values coincide: we then let
  $\Val(G,f) = \uppervalue(G,f) = \lowervalue(G,f)$ be the value of
  the game.
\end{proposition}
\begin{proof} We rely on Martin's determinacy theorem for Blackwell
  games~\cite{Mar98}, since all the cost functions considered are
  Borel measurable.
\end{proof}

\subsection{Encoding of quantitative sabotage games}

Starting from a \QSG $\game = (V,E,B,v_I,\distr_I,f)$, we encode it in
the quantitative two-player game
$\sem\game = ((\VMinSem,\VMaxSem, \allowbreak \edgesSem,
\wSem,\InitSem),f'_\wSem)$ as follows:
\begin{itemize}
\item $\VMinSem = V \times \Delta(\game)$,
  $\VMaxSem= E \times \Delta(\game)$: \Min vertices represent
  configurations of \game (i.e., the vertex of \game currently
  occupied by \runner, together with the current budget distribution),
  and \Max vertices encode the last edge played by \runner in \game
  and again the current budget distribution;
\item
  $\edgesSem = \{ \big ( (v,\distr) , (e,\distr) \big ) \st
  e=(v,v')\in \edges\} \cup \{ \big ( (e,\distr) , (v',\distr') \big )
  \st e=(v,v')\wedge \distr \buRel \distr'\}$;
\item for all $e=(v,v')\in \edges$ and $\distr,\distr'\in \Delta(\game)$ we let
  $\wSem\big((v,\distr), (e,\distr)\big)=\distr(e)$ and
  \begin{itemize}
  \item $\wSem\big((e,\distr),  (v',\distr') \big)=0$ if $f=\discfun{\lambda}$,
  \item $\wSem\big((e,\distr),  (v',\distr') \big)=\distr(e)$ otherwise;
  \end{itemize}
\item $\InitSem = (v_I,\distr_I)$ is the initial configuration;
\item if $f=\discfun{\lambda}$, we let $f'=\discfun{\sqrt{\lambda}}$,
  otherwise $f'=f$.
\end{itemize}

We claim that $\game$ and $\sem\game$ are equivalent, meaning that
they have the same value. The main difference lies in the way costs
are computed. Indeed, consider a pair of consecutive moves from both
players in the original \QSG $\game$, i.e., the traversal of an edge
$e=(v,v')$ by \runner, followed by a budget redistribution
$\distr\buRel\distr'$ by \saboteur. Observe that this pair of moves
incurs a cost of $\delta(e)$ in the original \QSG, but is encoded by
the traversal of \emph{two consecutive edges}
$\big((v,\distr), (e,\distr)\big)$ and
$\big((e,\distr), (v',\distr') \big)$ in $\gameSem$ that have
\emph{both} weight $\delta(e)$ (or weight $\delta(e)$ and then weight
$0$ for the discounted sum case). Observe however that this is not a
problem for the cost functions that we are considering.  Indeed,
$\supfun$, $\inffun$, $\lsupfun$, and $\linffun$ are resistant to
stuttering. The value of the average cost is also consistent, since
both the sum of the visited weights \emph{and} the length of the paths
are doubled in $\gameSem$ with respect to $\game$.  For the discounted
sum, this is taken care of by replacing the original discount factor
$\lambda$ in $\game$ by $\sqrt{\lambda}$ in $\gameSem$.

\section{Proofs of Section~\ref{sec:solving}}

\subsection{\texorpdfstring{Reduction from ABF to $\ESPr$: proof of
    Lemma~\ref{lem:abf-to-extended}}{{Reduction from ABF to ESPr:
      proof of Lemma~\ref{lem:abf-to-extended}}}}
\label{sec:ABF2ESPr}

In this section, we fix an instance of the ABF problem, i.e., a CNF
formula $\Phi$ and an initial configuration
$(\valuation_0, \nplayer_0)$.  We let $N$ be the number of variables
in $\Phi$. We construct an extended \QSG $\game$ such that
\saboteur wins in the extended safety problem over $\game$ if and only
if Prover wins the ABF game. Therefore, \saboteur will act as Prover
while \runner will act as Disprover. As an example of our
construction, we consider the CNF formula of Example~\ref{exa:abf}.
Recall that in extended \QSGs, we allow for the use of safe edges,
i.e., edges where \saboteur cannot put budget, and final vertices,
i.e., vertices where the play ends, with the budget placed on this
vertex taken into account to compute the cost. We present step by step
the vertices contained in $\game$. For every variable $X$, we create
$4$ final vertices
$\Ver(X) = \{\neg X^{(1)}, \neg X^{(2)}, X^{(1)}, X^{(2)}\}$.  We also
create another final vertex called $\alpha$. In the following, we
always assume that edges are safe, unless explicitly stated.

\subsubsection*{Forcing to have at least budget 2 on $\Ver(X)$}
For each variable $X$, and each triplet
$t= \{v_1,v_2,v_3\} \subset \Ver(X)$, we create two vertices,
$t^{(1)}$ and $t^{(2)}$ such that in the graph,
$(t^{(1)}, t^{(2)}), (t^{(2)}, v_i)\in E$ for all $i\in \{1,2,3\}$.
Note that if there is zero budget on the triplet $t$ when \runner
arrives in $t^{(1)}$, then \runner is sure to reach one of the $v_i$
without visiting edges with non-zero budget (and hence win the
game). For all vertices $v$ in the graph (except those in the
initialisation gadget, as we see later), all variables $X$ and all
triplets $t$ as described above, we create an edge $(v, t)$. If at
some point in the game there is a variable $X$ such that there are
less than $2$ vertices in $\Ver(X)$ with a budget on them, then
\runner is sure to win the game. This gadget is depicted in
\figurename~\ref{GadgetGeq2}. In the following, we assume that
\saboteur always place at least 2 units of budget on each
$\Ver(X)$. We also assume that if it is the case, then \runner does
not go on a $t^{(1)}$ vertex (indeed he will be sure to lose the play
if he does so).
\saboteur always places at least 2 units of budget on each $\Ver(X)$. We
also assume that if it is the case, then \runner does not go on a
$t^{(1)}$ vertex (indeed he will be sure to lose the play if he does
so).

The budget in the game is $B=2N+1$. Let $v$ be a vertex that has an
outgoing edge towards $\alpha$. When \runner leaves $v$, in order for
\saboteur not to lose, there must be one unit of budget on $\alpha$
and exactly 2 units of budget on each $\Ver(X)$.

\subsubsection*{Forcing the budget to be well distributed} Now we
present another gadget that allows \runner to force, on some vertices,
that $(\star)$ either $\neg X^{(1)}$ and $\neg X^{(2)}$ have one unit
of budget each, or $X^{(1)}$ and $X^{(2)}$ have one unit of budget
each. To do so, for each pair $p=\{\neg X^{(i)},X^{(j)}\}$, we
construct a vertex $p$ that has two outgoing edges, $(p,\neg X^{(i)})$
and $(p, X^{(j)})$. Let $v$ be a vertex that has outgoing edges toward
each of those $p$ vertices. When \runner leaves $v$, if property
$(\star)$ is not fulfilled, then there is a pair
$p=\{\neg X^{(i)},X^{(j)}\}$ with zero budget on it. By going on $p$,
\runner ensures to reach one of those vertices without budget, hence
to win the game. We let $Check(X)$ be the set of all those vertices
associated with $X$, and we will describe later which vertices have
outgoing edges toward $Check(X)$. This gadget is depicted in
\figurename~\ref{Check(x)}. When \runner leaves a vertex $v$ with an
outgoing edge toward $Check(X)$, we now assume that \saboteur has made
true property $(\star)$ for $X$, so that \runner never goes to
$Check(X)$ (indeed he will be sure to lose if he does so when
$(\star)$ is fulfilled).

Let $v$ be a vertex that is connected to $\alpha$ and to all vertices
of $Check(X)$ for all $X$ (we let
$Check=\bigcup_{X\in \textit{Var}} Check(X)$). When \runner leaves
$v$, in order for \saboteur not to lose, there must be one unit of
budget on $\alpha$, and one unit of budget either on $\neg X^{(1)}$
and $\neg X^{(2)}$, or on $ X^{(1)}$ and $X^{(2)}$, for all $X$. We
call such a configuration a valid one, and remark that there is an
immediate bijection from valid configurations and valuations of the
variables of the CNF formula. We call a \textit{valid vertex} a vertex
connected to $\alpha$ and to $Check$.

\subsubsection*{Initialising the game}\label{initialisation} We add a gadget, at the beginning
of the game, forcing \saboteur to distribute the budget accordingly to
the initial valuation $\valuation_0$ of the ABF game. The gadget works
as follows: \runner crosses $2N+1$ safe edges successively, and then
goes on a vertex $v$ that has edges towards each vertex of the
required initial configuration. \saboteur has the time to put the
required units of budget on this configuration, and is forced to do
so, otherwise \runner would be able to reach a final
vertex. Therefore, we are sure that once this gadget is left to start
playing the game, the configuration is indeed the required one. From
vertex $v$, there is also another safe edge going either to the vertex
$Play$ or to the vertex $set^{(1)}$ (the role of both vertices is
explained later), depending on whether $\nplayer_0$ is Disprover or
Prover, respectively.

\subsubsection*{Structure of the graph} From the CNF formula of
Example~\ref{exa:abf}, we construct the graph depicted in
\figurename~\ref{fig:exa-abfMain}. For the sake of clarity, we omit
the gadgets introduced above. Double bordered vertices represent final
vertices, and double arrows represent safe edges. As stated above,
from all vertices depicted here, gadget of
\figurename~\ref{GadgetGeq2} is used to check that $\Ver(X)$ contains
at least budget 2, for all variables $X$. The subscript in
$\{A,B,C,D,\alpha\}$ on vertices depicts an edge from the vertex to
the corresponding gadget $Check$, or vertex $\alpha$.

\subsubsection*{Saboteur modifies Prover's variables} The two safe
vertices $set^{(1)}$ and $set^{(2)}$ describe Prover's turn to modify
one of its variables. Both vertices have an outgoing edge towards
$\alpha$ ensuring that one pebble is left on it. $set^{(2)}$ is a
valid vertex, and $set^{(1)}$ is connected to $Check(X)$ for all
variables $X$ belonging to Disprover. Finally, there is an edge
$(set^{(1)}, set^{(2)})$ connecting those two vertices.

Let $v$ be a valid vertex with an outgoing edge $set^{(1)}$ and let
$\valuation_1$ be the valuation of variables induced by a valid
configuration at the moment \runner leaves~$v$. If \runner goes to the
vertex $set^{(1)}$ and then to $set^{(2)}$, let $\valuation_2$ be the
valuation induced by the valid configuration at the moment \runner
leaves $set^{(2)}$. We claim that between $\valuation_2$ and
$\valuation_1$, at most one variable of Prover has been
modified. Indeed after \runner has arrived in $set^{(1)}$, \saboteur
cannot remove the budget on $\alpha$, and he cannot take the budget
on some $\Ver(X)$ to put it on another $\Ver(X')$, with $X'\neq X$,
as there would be only budget 1 on $\Ver(X)$ and \runner would
win. Therefore, the only possible move for \saboteur is to
redistribute the budget inside some $\Ver(X)$. Moreover, if $X$
belongs to Disprover after a move, $\Ver(X)$ will not satisfy the
property $(\star)$ and, since $set^{(1)}$ is connected to $Check(X)$,
\runner would win. Therefore, either \saboteur does nothing, or he
redistributes the budget inside some $\Ver(X)$ where $X$ belongs to
Prover. If he has done nothing then after \runner has gone to
$set^{(2)}$, by the same reasoning, and by the necessity that at this
moment the configuration is valid, one can ensure that again \saboteur
does nothing, in which case we would have
$\valuation_2=\valuation_1$. Let us focus on the case where Prover has
performed some redistribution in $\Ver(X)$. Without loss of
generality, assume that when leaving $v$, the budget was placed on
$X^{(1)}$ and $X^{(2)}$, and after leaving $set^{(1)}$ the budget is
on $\neg X^{(1)}$ and $X^{(2)}$. By the same reasoning, we know that
after reaching $set^{(2)}$, \saboteur can only redistribute the budget
inside a $\Ver(X')$ where $X'$ belongs to Prover. Furthermore, if
$X'\neq X$ then, when leaving $set^{(2)}$, $\Ver(X)$ would not satisfy
$(\star)$ and the configuration would not be valid. Therefore
\saboteur can either choose to have the budget on $X^{(1)}$ and
$X^{(2)}$, or on $\neg X^{(1)}$ and $\neg X^{(2)}$, therefore between
$\valuation_2$ and $\valuation_1$ only the valuation of $X$ may have
change.

\subsubsection*{Runner modifies Disprover's variables} From the vertex
$Play$, \runner chooses a variable $X$ of Disprover, and goes either
to $set^{(1)}_{\neg X}$ or to $set^{(1)}_{ X}$: assume without loss of
generality that he goes to $set^{(1)}_{X}$. Those two vertices have
outgoing edges toward $\alpha$ and toward $Check(X')$ for all
$X'\neq X$. Let $\valuation_1$ be the valuation associated with the
valid configuration when \runner leaves $Play$. After arriving in
$set^{(1)}_{X}$, \saboteur can only redistribute the budget inside
$\Ver(X)$. After arriving in $set^{(2)}_{X}$, \saboteur is forced to
reach a valid valuation, therefore if he has modified the budget
distribution in $\Ver(X)$, he must do it again in order for $\Ver(X)$
to satisfy $(\star)$. Furthermore, as $set^{(2)}_{X}$ has outgoing
edges to the two final vertices $X^{(1)}$ and $X^{(2)}$, there must be
a unit of budget on each of those vertices. Therefore, if we let
$\valuation_2$ be the valuation induced by the valid configuration
when \runner leaves $set^{(2)}_{X}$, $\valuation_2$ must be equal to
$\valuation_1$ except possibly for $X$ that must now be true.

\subsubsection*{Verifying a valuation} Before explaining the whole
behaviour of the game, let us describe the verification process. As
$\textit{Verif}$ is a valid vertex, when \runner leaves this vertex, the
configuration is valid: we therefore let $\valuation$ be the valuation
induced by this configuration. We show here that, from the moment
\runner leaves $\textit{Verif}$, \saboteur has a winning strategy if and only
if $\valuation$ satisfies the CNF formula. Let us first describe this
part of the arena.

$\textit{Verif}$ has one outgoing safe edge toward each vertex
$\chi_i$ associated with the eponymous clause. Those vertices are
connected to $Check$. Take a clause
$\chi_i = at_1 \wedge \cdots \wedge at_\ell $. For each strict suffix
of this clause containing at least two atoms, i.e., for each
sub-clause of the form $at_j \wedge \cdots \wedge at_\ell$ with
$1<j<\ell$, create an eponymous vertex. Then $\chi_i$ has a safe edge
toward $at^{(1)}_1$ and a (non safe) edge toward the rest of the
clause, i.e., nothing if $\ell=1$, $at^{(1)}_2$ if $\ell =2$, and the
vertex `$at_2 \wedge \cdots \wedge at_\ell$' if $\ell>2$.  The same
principle applies to the vertex `$at_2 \wedge \cdots \wedge at_\ell$',
etc. For example, take $\chi_1$ in the CNF formula~$\Phi$. The vertex
$\chi_1$ has edges toward $\neg A^{(1)}$ and toward the vertex
`$B \vee \neg C$' which is the rest of the clause. Then the vertex
`$B \vee \neg C$' has an edge toward $B^{(1)}$ and an edge toward
$\neg C^{(1)}$.

Assume first that $\valuation$ satisfies the formula, and let us see
how \saboteur has a winning strategy. When \runner reaches a clause
$\chi_i$, we know that it is true in $\valuation$, i.e., that one of
its atom is true. On the game, this is represented by the fact that
one of the atoms $at_j$ has non-zero budget on the two associated
vertices $at^{(1)}_j$ and $at^{(2)}_j$. For example assume that
\runner goes to $\chi_1$ and that $B$ is true, i.e., there is some
budget on $B^{(1)}$ and $B^{(2)}$. \saboteur will use the budget on
$\alpha$ to guide \runner in direction of this atom. In the example,
when \runner reaches $\chi_1$, \saboteur will put the budget on
$A^{(1)}$, then when \runner will go to `$B\vee \neg C$', \saboteur
will move the same unit of budget on $\neg C^{(1)}$, forcing \runner
to go to $B^{(1)}$. However, as there was already some budget on
$B^{(1)}$, \runner cannot leave `$B\vee \neg C$' without touching
some non-zero budget, and loses the safety game.

On the other hand, assume that $\valuation$ does not satisfy the
formula and let us see how \runner has a winning strategy. As the
valuation does not satisfy the formula, there exists a clause $\chi_i$
that is false. \runner goes to this clause. As it is false, all the
atoms are false, in particular, in the game, for all $at_j$, there is
budget 0 on $at^{(1)}_j$. \runner will have the following
behaviour. If, after reaching $\chi_i$, \saboteur has not put some
budget on $at^{(1)}_1$, then he goes there and wins, otherwise he goes
to the vertex representing the rest of the formula. From there, the
same reasoning applies: if \saboteur has not put some budget on
$at^{(1)}_2$, then \runner goes there and wins, otherwise he reaches
the next sub-clause. At the end, \runner reaches the vertex
`$at^{(1)}_{\ell-1}\vee at^{(1)}_{\ell}$', and whatever \saboteur
does, \runner reaches a final vertex with budget 0.

\subsubsection*{How the game works} When \runner leaves vertex $Play$, the
configuration is valid; once he reaches $set^{(1)}$, the configuration
is valid again, and the difference with the previous one is that the
valuation may have changed for at most one variable.  Once reaching
$Choose$, \saboteur may also have changed the valuation of one of its
variables. When \runner reaches $Choose$, \saboteur can only
redistribute the budget on $\alpha$. One can easily see that he has no
interest in changing the valuation by putting some budget in $\Ver(X)$
for some variable $X$, as at the next step he must put the budget back
on $\alpha$. However, \saboteur can either put the free unit of budget
on the edge $(Choose, Play)$, forcing \runner to go to the
verification part on the game, or put it on the edge
$(Choose, \textit{Verif})$, forcing \runner to remain in the part of
the game where they change the valuation. If \saboteur has a winning
strategy in the ABF game, he will apply it, and once the valuation
satisfies the formula, he will force \runner to go to the verification
part. On the other hand, if the formula is never true, \saboteur is
forced to prevent \runner from going to the verification part
(otherwise \runner would reach a final vertex as seen above), and the
game will last forever, allowing \runner to win.

\begin{example}
  Consider the formula given in Example~\ref{exa:abf}, i.e.,
  $\Phi = \clause_1 \wedge \clause_2 \wedge \clause_3 \wedge
  \clause_4$
  where $\clause_1 = A \vee \neg C$, $\clause_2 = C\vee D$,
  $\clause_3 = C\vee \neg D$ and $\clause_4=B\vee \neg B$.  The
  $\ESPr$ constructed from $\Phi$ is given in
  Figure~\ref{fig:exa-abfMain}.  Notice that besides the variable
  vertices, there is one extra final vertex, $\alpha$. In this
  construction, Saboteur plays the role of Prover, whose variables are
  $C$ and $D$, and Runner the one of Disprover whose variables are $A$
  and $B$.

  For the sake of clarity, edges pointing towards $\alpha$, as well as
  the two gadgets of Figures~\ref{GadgetGeq2} and~\ref{Check(x)} are
  omitted.  Consider that from all vertices but the variable ones, one
  can check condition $(i)$ for all variables, i.e., in order not to
  lose, Saboteur maintain a non-zero budget on at least two vertices
  from $\Ver(x)$ for all variable $x$.  Furthermore, on the bottom
  right corner of nodes are written the variables for which one can
  check condition $(ii)$ and whether there is an outgoing edge
  pointing towards $\alpha$, e.g., when Runner is in vertex
  $\textit{set}^{(1)}$, Saboteur must ensure that $A$ and $B$ satisfy
  condition $(ii)$ and that there is a non-zero budget on $\alpha$. In
  the following, we consider those gadgets as constraints, considering
  that condition~$(i)$ always holds, and for example that if Saboteur
  is in $\textit{set}^{(1)}$ we are sure that $(ii)$ holds in $A$ and
  $B$ and that there is a non-zero budget in $\alpha$.

  If we let $n$ be the number of variables (here $n=4$), let us set
  the budget to $2n+1 = 9$. In this context, each $\Ver(x)$ contains
  $2$ units of budgets, and the remaining unit can be either on
  $\alpha$, on the outgoing edge of $\textit{Choose}$, or on one of
  the variable vertices.
 
  The initialisation gadget ensures that after some preliminary steps,
  Runner reaches vertex $\textit{Play}$, and there is one unit on
  $\alpha$, and for each variable there are exactly two units of
  budget either on $\{\neg x^{(1)}, \neg x^{(2)}\}$ or on
  $\{x^{(1)}, x^{(2)}\}$, depending on the initial configuration of
  the ABF game.
 
  Let us now focus on the upper part of the game. When Runner is on
  vertex $\textit{Play}$, condition $(ii)$ must be satisfied for all
  vertices, and there must be one unit of budget on $\alpha$,
  therefore the budget describes a valuation of the variables, e.g.,
  on $\Ver(A)$ either the two units of budget are on
  $\{\neg x^{(1)}, \neg x^{(2)}\}$ in which case we consider that $A$
  is false, or on $\{x^{(1)}, x^{(2)}\}$ in which case $A$ is
  true. Assume that $A$ is false, and Runner wants to change its
  valuation. Then, he goes to $\textit{set}^{(1)}_A$ where Saboteur
  has the possibility to move one unit of budget in $\Ver(A)$, and
  then he goes to $\textit{set}^{(2)}_A$. In this configuration
  condition $(ii)$ must be satisfied for $A$. Furthermore if the two
  units of budget are still on $\{\neg A^{(1)}, \neg A^{(2)}\}$, then
  Runner wins by going on $A^{(1)}$, thus Saboteur has been force to
  switch the two weights on $\{A^{(1)}, A^{(2)}\}$. Then, a similar
  process allows Saboteur to modify the valuation of one of its
  variables, when Runner goes through $\textit{set}^{(1)}$ and
  $\textit{set}^{(2)}$. Those steps simulate one round of the ABF
  game.
 
  On vertex $\textit{Choose}$, Saboteur may remove the budget on
  $\alpha$ and put it on one of the outgoing edges of Choose, thus he
  can force Runner to go either on $\textit{Play}$ or on
  $\textit{Verif}$. If $\textit{Play}$ is chosen, both players will
  simulate another round of the ABF game. If it is $\textit{Verif}$,
  then Runner goes to the lower part of the game.
 
  In this part, Runner chooses a clause and then Saboteur can move the
  unit of budget that were on $\alpha$. For example, assume that
  Runner chooses $\clause_1$. As there were a unit of budget on
  $A^{(1)}$, Saboteur can take the budget of $\alpha$ to put in on
  $\neg C^{(1)}$ ensuring to win. Observe that the verification part
  of the game ensures that Saboteur wins if, for each clause, at least
  one of the atoms is true. Indeed if it is the case, whatever clause
  is chosen by Runner, Saboteur will be able, as seen above, to
  prevent Runner to play. On the other hand, if there is a clause
  where both atom are false, it means than both outgoing edges point
  towards empty final vertices, therefore whatever Saboteur does on
  the next step, Runner will be able to reach one of them, and thus
  win the game.
\end{example}

\subsection{\texorpdfstring{Reduction from \ESPr to \SPr: proof of
    Lemma~\ref{lem:safe-extended-wlog}}{Reduction from ESPr to SPr:
    proof of Lemma~\ref{lem:safe-extended-wlog}}}

We describe how to transform an extended \QSG into a regular \QSG. The
transformation rids the original sabotage game of its safe edges and
final vertices, and replaces them with corresponding gadgets with the
same properties.
	
Final vertices are replaced by the gadget shown in
\figurename~\ref{finalNodes}.  More formally, all edges incident in a
final vertex are replaced by edges incident on a copy of the
gadget. $A$ is the entry point of the gadget, i.e., any edge pointing
towards the final vertex in the extended \QSG would now lead to
$A$. Vertices $C_1$ and $C_2$ are both connected to $\alpha_i$, for
all $1 \le i \le B + 1$, and the $\alpha_i$'s form a clique of size
$B+1$.  It should be clear that, if \runner reaches one of the
$\alpha_i$, then he can ensure that the value of the play, from then
onwards, is exactly 0. Indeed, as there are $B+1$ outgoing edges, at
least one of them has no budget on it; if \runner crosses this edge,
he reaches another $\alpha_j$ where the same property holds. Thus, one
can easily see that when \runner reaches $A$, he can win if and only
if there is no budget on either one of the edges: $(A,C_1)$,
$(A,C_2)$.

Safe edges can be encoded as follows. Assume that we have a safe edge
$(A,C)$ in the extended \QSG. To encode it in a standard \QSG (with
final vertices, as we have already seen how to encode them), we add
$B+1$ vertices $E_1,\ldots,E_{B+1}$, and $B+1$ final vertices
$F_1,\ldots,F_{B+1}$. We remove the edge $(A,C)$, and add the edges
$(A,E_i)$, $(E_i,F_i)$ and $(E_i,C)$, for all $i\leq B+1$. The gadget
is depicted in \figurename~\ref{safeEdges}. \runner has a strategy to
go from $A$ to $C$ without crossing an edge with non-zero budget, and
forcing \saboteur to move at most one unit of budget inside the
game. That is to say, we have introduced one additional step to get
from $A$ to $C$, but we will see that \saboteur cannot move more than
one unit of budget on edges outside of the gadget, or he loses.
Indeed, when \runner leaves $A$, there must exist $i$ such that there
is no budget on edges $(A,E_i)$, $(E_i,F_i)$, $(E_i,C)$ nor on the
final vertex $F_i$. If \runner goes to $E_i$, \saboteur must take a
unit of budget and put it either on $(E_i,F_i)$ or on $F_i$, otherwise
\runner can reach $F_i$ and win. Now, \runner is able to reach $C$,
and then \saboteur can redistribute the budget as he wants.

\subsection{\texorpdfstring{Reduction from $\ThPr_{lsupfun}(0)$ to
    $\SPr$: proof of Lemma~\ref{lem:safe-to-limsup}}{Reduction from
    the threshold problem for LimSup to SPr: proof of
    Lemma~\ref{lem:safe-to-limsup}}}

Consider an instance $\mathcal{I}$ of the safety problem with
underlying graph $G(\mathcal{I}) = (V,E)$, budget $B$, and a starting
vertex $v_I$.
We build a \QSG $\mathcal G'$ with graph $G'=(V',E')$, initial vertex
$t_1$, budget $B'$, and cost function $\lsupfun$ as follows:
\begin{align*}
  B' &= |E|\,, \\
  V' &= V\cup \{s_i^j \st 1 \leq i \leq B' + 1, 1 \leq j \leq B'\} \cup \{x_m \st B+1 \leq m \leq B'\} \cup \{t_1,t_2\}\\
     &\phantom{{}= V} \cup \{f_k \st 1 \leq k \leq B\} \cup \{e_\ell \st 1 \leq \ell \leq B + 1\}\,, \\
  E' &= E \cup \{(e_i,f_j) \st 1 \leq i \leq B + 1, 1 \leq j \leq B\} \cup \{(f_k,e_k), (f_k,e_{k+1}) \st 1 \leq k \leq B\} \\
     &\phantom{{}= E}     \cup \{(x_\ell,t_m) \st 1 \leq m \leq 2, B+1 \leq \ell \leq B'\} \cup \{(t_m,s_i^{B'}) \st 1 \leq i \leq B' + 1, 1 \leq m \leq 2\} \\
     &\phantom{{}= E}     \cup \{(s_i^j,s_{i'}^{j-1}) \st 1 \leq i,i' \leq B' + 1, 2 \leq j \leq B'\} \cup \{(s_i^1,e_m) \st 1 \leq m \leq 2, 1 \leq i \leq B+1\} \\
     &\phantom{{}= E}     \cup \{(e_m,u) \st 1 \leq m \leq B+1,  (v_I,u) \in E\} \cup \{(u,x_\ell) \st u \in V, B+1 \leq \ell \leq B'\} \\
     & \phantom{{}= E}    \cup \{(e_m,x_\ell) \st 1 \leq m \leq B+1, B+1 \leq \ell \leq B'\}\,.
\end{align*}
Intuitively, the sub-graph of $G'$ defined by the vertices $e_i$ and
$f_j$ form an initial gadget which ensures that \runner can stay out
of $G(\mathcal{I})$ without paying, as long as there is some weight
assigned to edges from $E$. We also add an exit gadget consisting of
the sub-graph of $G'$ defined by the $x_k$ vertices. These allow
\runner to exit from $G(\mathcal{I})$ if \saboteur ``cheats'' by
assigning more weights to edges from $E$ than the original bound $B$.
Both gadgets are linked by a ``safe path'' formed by the vertices
$s_i^j$. Note that we add sufficiently many $s_i^j$ so that, for
\runner, getting from any $s_i^{B'}$ to any $s_j^1$ is always possible
without traversing a weighted edge.

We prove that \runner wins in $\mathcal I$ if and only if
$\Val(\mathcal G') \leq 0$.

Assume first \runner wins $\mathcal I$.  In $\mathcal G'$, he has no
trouble following a path from $t_1$ through the $u_i^j$ until he
arrives on some $u_i^1$ with budget distribution $w_0$ such that
$w_0(u_i^1,e_j) = 0$, for some $1 \le j \le B + 1$, since there are
$B'+1$ vertices at each level of the safe path. On his next turn, he
can then move to such an $e_j$. As long as the budget distribution has
some budget assigned to some edge of $E$, there exists a vertex $e_k$
or $x_{\ell}$ with no budget on either in-edges or out-edges,
respectively. In the first case, \runner can go to such such an $e_k$
via $f_k$ without paying anything. In the second case, \runner can get
to $t_1$ or $t_2$ via $x_\ell$ and repeat the process, all without
paying.
When the budget distribution has no weight assigned to edges of $E$,
\runner can follow his strategy from $\mathcal I$ -- with the
exception that he plays his first move from $e_j$ instead of $v_I$ --
as long as \saboteur keeps at most $B$ budget units on edges of
$E$. When this is no longer the case, say \runner is on a vertex $u$,
with budget distribution $w_1$, that means there are at most $B'-B-1$
budget units on other edges, hence there is a vertex $x_\ell$ such
that $w_1(u,x_\ell) = w_1(x_\ell,t_1) = w_1(x_\ell,t_2) = 0$.  \runner
then moves to $x_\ell$.  On his next turn, he can then move to either
$t_1$ or $t_2$, following an edge with no weight on it.  Then \runner
can restart this strategy.

Assume now, that \saboteur wins $\mathcal I$. From the start of the
game, \runner will have to traverse one $s_i^j$ for all $j$ from $B'$
to $1$.  When \runner is on a vertex $s_i^j$ for $j$ between $B+1$ and
$B'$, \saboteur puts a budget unit on the edge $(x_j,t_1)$ and leaves
it there. Similarly, when \runner is on a vertex $u_i^j$ for $j$
between $1$ and $B$, \saboteur puts a unit of budget back on the edge
$(f_j,e_j)$ and leaves it there.  When \runner finally reaches some
$e_k$, \saboteur passes.  Then, if \runner goes to $f_\ell$ or $x_m$,
\saboteur can assign some budget to $(f_\ell,e_{\ell + 1})$ or
$(x_m,t_2)$ and put it back where it was after \runner's next move,
where he will inevitably cross a weighted edge, then wait until
\runner gets back to some $e_k$.
Alternatively, from $e_k$, \runner can move to a vertex in
$G(\mathcal{I})$. In this case, \saboteur follows
his strategy from $\mathcal I$, using budget units assigned to edges
of the form $(f_k,e_k)$ when needed, until \runner crosses an edge of
$E$ with some weight on it or gets to some $x_\ell$.  In the latter
case, \saboteur can react the same way as if \runner was coming from
$e_k$. In the former case, \saboteur can start putting some weights on
all edges of $E$ until \runner gets to some $x_\ell$.  If \runner
never does, he will pay one at each step, which is enough for
\saboteur.  Otherwise, \runner goes to some $x_\ell$, then to $t_1$ or
$t_2$, where \saboteur can restart his strategy.

\section{Proofs of Section~\ref{sec:static}}

\subsection{\texorpdfstring{Static threshold problem for $\inffun$ and
    $\linffun$ is in \P: proof of
    Lemma~\ref{lem:poly-inf-linf}}{Static threshold problem for Inf
    and LimInf is in PTIME: proof of Lemma~\ref{lem:poly-inf-linf}}}

For $\inffun$, we claim that
$\Valstatic(\game)=\lfloor |\overline E|/B \rfloor$, where
$\overline E$ is the set of edges reachable from $v_I$.  Indeed, for a
given budget distribution $\distr$, \runner simply goes towards the
edge reachable from $v_I$ with the least budget possible; therefore,
\saboteur must place \emph{equal} budget on each such edge. With a
budget $B$, he can ensure $\lfloor |\overline E|/B \rfloor$ on every
edge (some edges may contain a bigger portion of the budget, but some
edges will always have at most $\lfloor |\overline E|/B \rfloor$).
Hence, deciding the static threshold problem for $\inffun$ amounts to
computing the set $\overline E$ (can be done in linear time with a
depth-first-search algorithm), and checking whether
$|\overline E|\leq B \times (T+1)$.

For $\linffun$, we must refine the study by considering strongly
connected components. Precisely, we claim that
$\Valstatic(\game)=\lfloor |\widetilde E|/B \rfloor$, where
$\widetilde E$ is the set of edges reachable from $v_I$ and contained
in a strongly connected component of the graph.  Indeed, for a given
budget distribution $\distr$, \runner simply goes towards a cycle
reachable from $v_I$ containing an edge with the least budget $b$
possible: he will visit infinitely often this edge, ensuring an
inferior limit at most $b$. Such a cycle is included in a strongly
connected component, and reciprocally, every edge of a strongly
connected component is part of a cycle. Hence, \saboteur must secure
\emph{equal} budget on each edge of every strongly connected
components. Then, deciding the static threshold problem for $\linffun$
amounts to computing the set $\widetilde E$ (can be done in linear
time, e.g., with Tarjan's algorithm), and checking whether
$|\widetilde E|\leq B \times (T+1)$.

\subsection{\texorpdfstring{Static threshold problem for $\supfun$,
    $\lsupfun$, $\mpfun$ and $\dsfun$ is \coNP-complete: proof of
    Lemma~\ref{lem:coNP-complete}}{Static threshold problem for Sup,
    LimSup, Avg and DS is coNP-complete: proof of
    Lemma~\ref{lem:coNP-complete}}}\label{app:coNP}

For the membership in $\NP$, we can first guess a budget distribution
$\distr$ (that is of size polynomial), and then compute the value of
the one-player (since player \Max has no choices anymore) quantitative
game $\gameStat{\distr}$, to check if it is greater than~$T$:
computing the value of such a game 
can be done in polynomial time for the four cost functions we
consider (see~\cite{ag11}).

To prove the $\NP$-hardness for cost functions $\lsupfun$ and
$\mpfun$, we give a reduction from the following problem. The \textsc{Feedback
arc set problem} consists in, given as input a directed graph $G =
(\vertices,\edges)$ and a threshold $k \leq |\edges|$, determining whether there
is a set $E'$ of $k$ edges of $G$ such that $(\vertices,\edges \setminus E')$ is
acyclic.
Karp showed in \cite{Kar72} that the feedback arc set problem is
$\NP$-complete.

We now use the feedback arc set problem to prove the results of
\coNP-hardness of the static threshold problem. Let us consider an
instance of the feedback arc set problem, given by a directed graph
$G = (\vertices,\edges)$ and a natural integer $k\le |\edges|$. We
suppose, without loss of generality, the existence of a vertex $v_I$,
without any in-going edges, and linked with an edge to every other
vertex: since $v_I$ is not included in any cycle, the set $E'$ of the
output of the problem has no interest at containing any of the edges
added in this way. 

We then construct a \QSG $\game = (\vertices, \edges,k,v_I,f)$ with
$f\in\{\lsupfun,\mpfun\}$. It is not difficult to show that
$\Valstatic(\game)>0$ if and only if there exists a set $E'$ of $k$
edges of $G$ such that $(\vertices,\edges\setminus E')$ is
acyclic. Indeed, $\Valstatic(\game)>0$ implies that there exists a
distribution $\distr\in\Delta(\game)$ such that for all strategies
$\rho$ of \runner, $f(\overline{\pi^\distr_{\rho,\iota}})>0$. Noticing
that every vertex is reachable from the initial vertex $v_I$, and
considering memoryless strategies of \runner (such that
$\pi^\distr_{\rho,\iota}$ ends with a simple cycle of the graph), we
show that every cycle contains at least one edge with a non-zero
budget. The set $E'=\{e\in E\mid \distr(e)>0\}$ is then a valid output
for the feedback arc set problem. For the reciprocal implication, we
simply assign a budget $1$ to each vertex of the set $E'$.

The result for $\supfun$ and $\dsfun$ is then obtained by a slight
modification of the previous proof. Let
$\underline\game = (\underline V,\underline E,k,v_I,\supfun)$ be the
\QSG obtained from $\game$ by transforming every edge $(v_I,v)$ into a
safe edge (see Lemma~\ref{lem:safe-extended-wlog}). Without loss of
generality, we can now assume that $\distr$ never assigns budget to
the edges $\underline E \setminus E$. We also note that $v_I$ has no
in-going edges so that every play in $\underline\game$ traverses a
safe edge at most once. We claim that $\Valstatic(\underline\game)>0$
if and only if there exists a set $E'$ of $k$ edges of $G$ such that
$(\vertices,\edges\setminus E')$ is acyclic. Indeed, if
$\Valstatic(\underline\game) > 0$, considering
$E'=\{e\in E\mid \distr(e)>0\}$, it is easy to show that
$(\vertices,\edges\setminus E')$ is acyclic: if not, \runner may
simply jump from $v_I$, with a safe edge, to one of the vertices of a
cycle of $(\vertices,\edges\setminus E')$, and then loop in this cycle
forever, without visiting any edge with non-zero budget. For the
reciprocal implication, again, it suffices to assign a budget $1$ to
each vertex of $E'$. The result for $\discfun{\lambda}$ follows from
the same reduction together with Lemma~\ref{lem:sup-to-ds}.

\section{Towards more expressive sabotage
  games}\label{sec:more-expressive}

In this section we increase the expressiveness of the definition of
sabotage games, and show that the threshold problem for these new
games are still in \EXP. The lower bound is immediate since they are
extensions of previous problems shown \EXP-hard in the rest of the
article.

One can see a sabotage game as a system in which a controller tries to
evolve while avoiding as much as possible the weights put by
Saboteur. The vertices of the graph represent configurations of the
system, edges represent the actions, and the budget of the Saboteur
may represent several problems that can occur during the
execution. For example, it may describe a number of failures that can
happen at the same time, or in a much quantitative way, it may
represent how much some elements of the systems are overload, and then
how much it would cost, in terms of time or energy, to use them.

We propose to look at sabotage as a particular semantics of systems.
Based on the observation of Appendix~\ref{sab2games}, remember that
one can define the semantics of a \QSG
$\game = (V,E,B,v_I,\distr_I,f)$ as a quantitative two-player game
$\sem\game$. If we split the model (the graph $G=(V,E)$ with initial
vertex $v_I$), from the sabotage parameters (budget $B>0$, initial
distribution $\distr_I\in\Delta(E,B)$, and cost function $f$), we can
define:
\[\sem{G}_{B,\distr_I,f} = \sem{(V,E,B,v_I,\distr_I,f)}\,.\]
We have seen that the value of the \QSG $(V,E,B,v_I,\distr_I,f)$ is
identical to the value of the quantitative two-player game
$\sem{G}_{B,\distr_I,f}$.

From a model point of view, graphs---which can be viewed as one-player
games with trivial winning conditions---are quite limited. In more
realistic models, we may be interested as modelling systems with
uncontrollable actions (i.e., as two-player games), and where the
controller has a specific Boolean goal to achieve, instead of simply
visiting the graph \emph{ad vitam \ae{}ternam}. A more realistic goal
is usually expressed via LTL formulas, that can be modelled into
qualitative games with parity winning conditions, as we show in the
following.

\subsection{Qualitative two-player games}
As a complement of the quantitative two-player games defined in
Appendix~\ref{sec:qgames-sem}, we now focus on \emph{qualitative}
two-player games games. Consider a weighted arena
$G=(V_\Min,V_\Max,E,w,v_I)$ as before. In the qualitative setting, we
are no longer interested in associating a value to each play (in
particular, the weight function $w$ is of no use here), but simply
stating whether a play is winning or not for a player. Formally, a
\emph{winning condition} is a subset of $V^\omega$ containing the set
of winning plays. A qualitative game is a pair $(G,C)$ consisting of
an arena $G$ and a winning condition $C$. A play $\pi$ is declared
winning for $\Min$ (respectively, for $\Max$) if $\pi\in C$
(respectively, $\pi\not\in C$). A strategy $\sigma_p$ of player $p$ is
winning for $p$ if all plays $\pi\in \outcomes(\sigma_p)$ are winning;
a play/strategy is losing for player $p$ otherwise. We say that player
$p$ wins (respectively, loses) the game if he has (respectively, does
not have) a winning strategy. Here are some usual winning conditions
considered widely in the literature:
\begin{itemize}
\item for all $F\in V$ or $F\in E$, $\textit{Reach}(F)$ is the set of
  plays that contain an occurrence of $F$.
\item for all $F\in V$ or $F\in E$, $\textit{Safe}(F)$ is the set of
  plays that do not contain any occurrence of $F$.
\item for all $F\in V$ or $F\in E$, $\textit{B\"uchi}(F)$ is the set
  of plays that contain infinitely many occurrences of $F$.
\item for all $F\in V$ or $F\in E$, $\textit{coB\"uchi}(F)$ is the set
  of plays that contain only finitely many (possibly none) occurrences
  of $F$.
\item for all $\Col\colon V \to \N$ (such mapping is called a
  \emph{colouring function}), $\textit{Parity}(\Col)$ is the set of
  plays $v_0v_1\cdots$ such that the greatest colour appearing
  infinitely often in the sequence $\Col(v_0), \Col(v_1),\ldots$ is
  even. Given a colouring function $\Col$, we let $|\Col|$ be the
  number of different colours of the vertices, i.e.,
  $|\Col| = | \{ \Col(v)\mid v\in V\}|$.
\item for all value function $f$, and $T\in \R$,
  $\textit{Threshold}^\leq(f,T)$ is the set of plays $\pi$ such that
  $f(\pi) \leq T$.
\end{itemize}

\begin{proposition}
  Qualitative two-player games with all winning conditions considered
  above are determined, i.e., one player is winning if and only if his
  opponent is losing.
\end{proposition}
\begin{proof}
  Martin's determinacy theorem~\cite{Mar75} applies here since all the
  above mentioned objectives are Borel sets.
\end{proof}

\subsection{Sabotage in  parity games}
In order to apply a sabotage semantics to qualitative game, where
$\Min$ wants to satisfy a condition while minimising a cost, one must
study some mixture between qualitative and quantitative aspects.  We
see how one can combine winning conditions and value functions, as
introduced in \cite{CHJ05}. Intuitively, in a weighted arena
$G=(V_\Min,V_\Max,E,w,v_I)$ with a winning condition $C$ and a value
function $f$, $\Min$ could want to satisfy $C$ while minimising
$f$. We formalise this by building a new value function, denoted by
$C\wedge f$, and defined by $C\wedge f (\pi) = +\infty$ if
$\pi\not\in C$, and $C\wedge f (\pi) = f(\pi)$ otherwise. The
quantitative two-player game $(G,C\wedge f)$ now contains the
combination of both objectives.

We may finally introduce a sabotage semantics for parity games.
Instead of deciding whether a player has a winning strategy, which
would be a standard semantics, we decide whether he has a winning
strategy that guarantees (or simply avoids in the case of a threshold
$0$) a certain threshold over the quantity of penalties when the game
is subject to failures.

Formally, given a two-player parity game
$G= ( (V_\Min,V_\Max, E, v_I), \Parity(\Col) )$, a budget $B>0$, an
initial distribution $\distr_I\in\Delta(E,B)$, and a cost function
$f\in\{\inffun,\linffun,\supfun,\lsupfun,\mpfun\}$, the
\emph{$B,\distr_I,f$-sabotage semantics of $G$} is the quantitative
game
\[\sem{G}_{B,\distr_I,f}= ((V'_\Min, V'_\Max, E', w,\InitSem), \Parity(\Col') \wedge f_w
)\,,\] where :
\begin{itemize}
\item $V'_\Min = V_\Min \times \Delta(E,B)$, and
  $V'_\Max = ( V_\Max \times \Delta(E,B) ) \uplus (E \times
  \Delta(E,B))$:
  with respect to the one-player case of
  Appendix~\ref{sec:qgames-sem}, we add some vertices to player \Max
  that has now in charge both the moves of the environment
  (uncontrollable actions), and redistributions of \saboteur;
\item
  $E' = \{ \big ( (v,\delta),(e,\delta) \big ) \mid e=(v,v')\in E\land
  \delta \in \Delta(E,B)\} \cup \{ \big ( (e,\delta), (v', \delta')
  \big )\mid e=(v,v')\in E\land \delta \buRel \delta'\}$;
\item
  $w \big( (v,\delta),(e, \delta)\big ) = w \big ( (e,\delta),(v',
  \delta')\big ) = \delta(e)$;
\item $\InitSem=(v_I,\distr_I)$ is the initial configuration;
\item
  $\Col' \big ( (v,v'), \delta \big ) = \Col' ( v',\delta ) =
  \Col(v')$.
\end{itemize}
To simplify our study, we do not consider the discounted-sum in this
section. 
The threshold problem, describing the cost that player $\Min$ can
ensure, is then defined as previously.

\begin{definition}[Threshold problem for cost function $f$]~\\
  \textsc{Input:} A parity game $(G,\Parity(\Col))$, a budget $B$, an
  initial distribution $\distr_I\in\Delta(E,B)$, and a threshold $T$,\\
  \textsc{Output:} Is there a strategy $\sigma$ of $\Min$ such that
  $\Val(\sem{G}_{B,\distr_I,f},\sigma)\leq T$?
\end{definition}

We are able to show that, even with the extension, the threshold
problem stays in \EXP.
\begin{theorem}\label{ComplexityParitySabotage}
  The threshold problem for cost functions $\inffun$, $\linffun$,
  $\supfun$, $\lsupfun$ and $\mpfun$ is in $\EXP$.
\end{theorem}

To prove this theorem, we first establish a crude (but sufficient)
upper bound on the complexity of solving quantitative games obtained
by combining parity winning conditions and the previous cost
functions.

\begin{proposition}\label{SolvingQuantPar}
  There exists three polynomials $p_1,p_2,p_3$ such that we can decide
  the threshold problem of any quantitative game
  $(G,\Parity(\Col)\wedge f_w)$ with
  $f\in \{\inffun,\linffun,\supfun, \lsupfun,\mpfun\}$ with a
  complexity in $O\big(p_1(|w|)\cdot p_2(|V|)^{p_3(|\Col|)}\big)$.
\end{proposition}
\begin{proof} We start with the case $f=\mpfun$. In~\cite{Z98,ag11} it
  has been shown that one can decide who wins in a qualitative game
  with a parity condition with a complexity in $O(|V|^{2+|\Col|})$.
  In~\cite{ZwiPat96}, it has been shown that one can compute the value
  of a quantitative game with an average cost function with a
  complexity in $O(|w|\cdot|V|^5)$. The combination has been studied
  thoroughly in~\cite{CHJ05}. There, it has been shown that if one can
  solve average cost games in $c_1(|V|,|w|)$ and parity games in
  $c_2(|V|,|\Col|)$, then one can solve games
  $(G,\Parity(\Col)\wedge \mpfun_w)$ with a complexity in
  $O(|V|^{|\Col|}(|V|^2+c_1(|V|,|w|)+c_2(|V|,|\Col|))$. By combining
  this result with the two above, we obtain a complexity in
  $O((|w|+1) |V|^{5+|\Col|})$.

  We then turn to the case
  $f\in\{ \inffun,\supfun,\linffun,\lsupfun\}$. Our proof goes by
  encoding $f$ into a qualitative winning condition, and then using
  classical results of algorithmic game theory. Observe that deciding
  the threshold problem in the game $(G,\Parity(\Col)\wedge f_w)$
  amounts to solving the following problem:
  \begin{center}
    \begin{minipage}{0.9\linewidth}
      \textsc{Input:} A weighted arena $G=(V_\Min,V_\Max,E,w,v_I)$, a
      colouring function $\Col$, a threshold $T\in \mathbf Q$\\
      \textsc{Output:} Does $\Min$ have a winning strategy in
      $(G,\Parity(\Col)\cap \textit{Threshold}^\leq(f_w,T))$?
    \end{minipage}
  \end{center}
  
  The crucial remark is that, if we let
  $F=\{ e\in E \mid w(e) \leq T\}$, we can rewrite the threshold sets
  for all payoff functions as follows:
  \begin{align*}
    &\textit{Threshold}^\leq(\inffun_w,T) = \textit{Reach}(F)\,, 
    && \textit{Threshold}^\leq(\linffun_w,T) = \textit{B\"uchi}(F)\,,\\
    &\textit{Threshold}^\leq(\supfun_w,T) = \textit{Safe}(E\setminus
      F)\,, 
    && \textit{Threshold}^\leq(\lsupfun_w,T) = \textit{coB\"uchi}(E
       \setminus F)\,.
  \end{align*}

  Notice that $F$ is a subset of edges, and not vertices. However, it
  is easy to transform the problem into an equivalent problem where
  $F$ is indeed a subset of edges. 
  Informally, it suffices to enrich the vertex set by letting
  $V'= V\times\{0,1\}$, letting $(v_I,0)$ the initial vertex instead
  of $v_I$, and replacing each edge $(v,v')\in E$ by the set of edges
  in $E'$:
  \begin{itemize}
  \item $\{((v,0),(v',0)), ((v,1),(v',0))\}$ if $(v,v')\not\in F$;
  \item $\{((v,0),(v',1)), ((v,1),(v',1))\}$ if $(v,v')\in F$.
  \end{itemize}
  Then, letting $F'= V\times \{1\}$ and
  $\Col'(v)=\Col'(v,0) = \Col'(v,1) = \Col(v)$, allows us to keep
  track, in the vertices, of whether the last seen edge is in $F$ or
  not.

  Now to conclude the proof, we describe two polynomials $p_1$ and
  $p_2$ such that deciding if \Min can win qualitative problem with a
  winning condition obtained by the intersection of a parity condition
  and another one from
  $\{\textit{Reach}(F),\textit{B\"uchi}(F),\textit{Safe}(F),\textit{coB\"uchi}(F)\}$
  with $F\subseteq V$ can be done with a complexity in
  $O(p_1(V)^{p_2(|\Col|)})$. We let $G=(\VMin,\VMax,E,v_I)$ the arena on
  which we play (the weight function is of no use anymore).

  For $\textit{Reach}(F)$ (respectively, $\textit{Safe}(F)$), one can
  construct in polynomial time a parity game
  $(G'=(\VMin',\VMax',\allowbreak E',v_I), \Parity(\Col'))$ such that
  $V\subseteq V'$, $|\Col'|\leq |\Col|+1$, and \Min wins in
  $(G', \Parity(\Col'))$ if and only if \Min wins in
  $(G,\Parity(\Col)\cap \textit{Reach}(F) )$ (respectively,
  $(G,\Parity(\Col)\cap \textit{Safe}(F))$). For $\textit{Safe}(F)$,
  we remove for each vertex $v\in F$ every outgoing edge in $F$, and
  add a self loop, colouring them with an odd colour. Hence, if a play
  reaches such a vertex the play is losing, and otherwise, it is
  winning if and only if the greatest colour seen infinitely often is
  even. For $\textit{Reach}(F)$, we create two copies of the game. In
  the first copy, every colour is odd, and for every vertex $v$ in $F$
  the outgoing edges are modified to go to the same target but in the
  second copy. In the second copy, nothing is changed. The play start
  in the first copy. In order to win, $\Min$ must go to the second
  copy (otherwise the colour will always be odd), i.e., must reach a
  vertex in $F$ and then the greatest colour seen infinitely often
  must be even. As we have seen above parity games can be solved
  $O(|V|^{2+|\Col'|})$ which concludes the proof for these cases.

  For $\textit{B\"uchi}(F)$ and $\textit{coB\"uchi}(F)$, there exist
  two colouring functions $\Col'$ and $\Col''$ such that:
  $|\Col'|=|\Col''|=2$, $\textit{B\"uchi}(F) = \Parity(\Col')$, and
  $\textit{coB\"uchi}(F) = \Parity(\Col'')$. Indeed, for $v\in F$,
  simply consider $\Col'(v)=2$ and $\Col''(v)=1$, and for $v\not\in F$
  consider $\Col'(v)=1$ and $\Col''(v)=0$. Therefore solving a game
  with a winning condition of the form
  $\Parity(\Col)\cap\textit{B\"uchi}(F)$ or
  $\Parity(\Col)\cap\textit{(co)B\"uchi}(F)$ can be turned into
  solving a game with a winning condition of the form
  $\Parity(\Col)\cap \Parity(\Col')$ with $|\Col'|\leq 2$. Such games
  have been studied in~\cite{CHP06}. They have shown that they can be
  solved with a complexity in $O(|V|+2)^{5(|\Col|+2)^2})$, which
  concludes the proof of the proposition. 
\end{proof}

We can finally establish the complexity of solving sabotage parity games.

\begin{proof}[Proof of Theorem~\ref{ComplexityParitySabotage}]
  From a parity game $(G,\Parity(\Col))$, a budget $B$, an initial
  distribution $\distr_I$, and a threshold $T$, one can construct
  $\sem{G}_{B,\distr_I,f}=((V'_\Min, V'_\Max, E',
  w,\InitSem), \Parity(\Col') \wedge f_w )$
  in exponential time. Proposition~\ref{SolvingQuantPar} shows that we
  can decide who wins from $(v_I,\distr_I)$ in this game with a
  complexity in $O(p_1(|w|)\cdot p_2(|V'|)^{p_3(|\Col'|)})$.

  We have $|w|=B$ and $|\Col'|=|\Col|$. Furthermore
  $|V'| \leq |E| \times |\Delta(B,E)|$.  As
  $|\Delta (B,E)| \leq B^{|E|}$ and $|E|\leq |V|^2$, we have
  $|V'| \leq |V|^{2}\times B^{|V|^2}$. Since $B$ is given in binary,
  we can suppose that $B$ is at most exponential in the size of the
  input of the problem, which, in summary, shows that we can solve the
  threshold problem in exponential time. 
\end{proof}

\subsection{Sabotage semantics on LTL games}
The linear temporal logic (LTL) is a logic whose formulas describe
properties of infinite sequences of predicate. More formally, given a
game arena $G$, a mapping $\textit{Pred}$ from vertices to a set of
predicate $P$ and an LTL formula $\phi$, the winning condition
$\phi(\textit{Pred})$ is the set of plays $v_0 v_1\cdots$ such that
the sequence $\textit{Pred}(v_0)\textit{Pred}(v_1) \cdots$ satisfies
$\phi$.

Solving LTL-games, with their standard semantics, is already
2-$\EXP$-complete~\cite{PR89}.  The 2-$\EXP$ membership can be
obtained by turning an LTL formula $\phi$ into a parity automaton
whose size is doubly exponential in the size of $\phi$, and solving
the parity game obtained by taking the product of the game arena with
the automaton.

When applying a sabotage semantics to an LTL game $\game$, we obtain a
game $\sem\game$ of size exponential in the initial arena, and whose
value function is a combination of a cost function and the LTL
formula.  By applying the same method as above, using the parity
automaton associated with the formula and taking the product of the
automaton with $\sem\game$, we obtain a game whose size is doubly
exponential in the size of $\game$, and whose value function is a
combination of a cost function and a parity objective. Applying the
above result, one can show that this game can be solved in 2-$\EXP$
with respect to the size of $\game$.

As the standard semantics is equivalent to a sabotage semantics with
budget $0$, the problem remains 2-$\EXP$-hard, and thus
2-$\EXP$-complete.

\end{document}